\documentclass[sigplan,10pt,authorversion,nonacm]{acmart}
\setcopyright{none}

\usepackage[utf8]{inputenc}
\usepackage[inline]{enumitem}
\usepackage{frenchineq}
\usepackage{amsmath}  
\usepackage{amssymb}     
\usepackage{amsthm}  
\usepackage{bbm}    
\usepackage{bm}
\usepackage{mathrsfs}
\usepackage{mathtools}
\usepackage{fixmath}
\usepackage{cleveref}

\usepackage{tikz}
\usetikzlibrary{calc}
\usetikzlibrary{arrows,arrows.meta}
\usetikzlibrary{positioning}

\usepackage{pgfplots}
\pgfplotsset{compat=1.18}

\definecolor{vertexone}{RGB}{0,114,178}
\definecolor{vertextwo}{RGB}{213,94,0}
\definecolor{vertexthree}{RGB}{0,158,115}
\definecolor{vertexfour}{RGB}{204,121,167}
\definecolor{vertexfive}{RGB}{230,159,0}
\definecolor{edgehighlight}{RGB}{117,112,179}

\definecolor{vplusone}{RGB}{230,159,0}
\definecolor{vplustwo}{RGB}{204,121,167}
\definecolor{vplusthree}{RGB}{86,180,233}
\definecolor{vminusone}{RGB}{0,114,178}
\definecolor{vminustwo}{RGB}{213,94,0}
\definecolor{vminusthree}{RGB}{0,158,115}

\usepackage[disable,colorinlistoftodos,bordercolor=orange,backgroundcolor=orange!20,linecolor=orange,textsize=tiny]{todonotes}
\allowdisplaybreaks

\newcommand{\ie}{i.e.}
\newcommand{\eg}{e.g.}
\newcommand{\etc}{etc}
\newcommand{\R}{\mathbb{R}}
\newcommand{\F}{\mathbb{F}}
\newcommand{\Ccal}{\mathcal{C}}
\newcommand{\Ksf}{\mathsf{K}}

\newcommand{\facets}{\Ksf^{(d-1)}}
\newcommand{\ridges}{\Ksf^{(d-2)}}
\newcommand{\eps}{\varepsilon}
\newcommand{\symdiff}{\mathbin{\triangle}}

\newcommand{\transpose}[1]{#1^\top}
\newcommand{\NP}{\mathsf{NP}}
\newcommand{\coNP}{\mathsf{coNP}}
\newcommand{\poly}{\mathrm{poly}}
\DeclarePairedDelimiterX{\scalar}[2]{\langle}{\rangle}{#1, #2}
\DeclareMathOperator{\vertices}{\mathrm{vert}}
\DeclareMathOperator{\linspan}{\mathrm{span}}
\DeclareMathOperator{\rank}{\mathrm{rank}}
\DeclareMathOperator{\conv}{\mathrm{conv}}
\DeclareMathOperator{\cone}{\mathrm{cone}}
\DeclareMathOperator{\interior}{\mathrm{int}}
\DeclareMathOperator{\relint}{\mathrm{relint}}
\DeclareMathOperator{\rec}{\mathrm{rec}}

\newcommand{\Rocq}{Rocq}

\newcommand{\cdd}{\C$cddlib$}
\newcommand{\lrs}{\C$lrslib$}
\newcommand{\MathComp}{Mathematical Components}
\newcommand{\CoqPolyhedra}{Coq-Polyhedra}

\usepackage{listings}

\let\C=\lstinline

\definecolor{dkblue}{rgb}{0,0.1,0.5}
\definecolor{lightblue}{rgb}{0,0.5,0.5}
\definecolor{dkgreen}{rgb}{0,0.4,0}
\definecolor{dk2green}{rgb}{0.4,0,0}
\definecolor{dkviolet}{HTML}{932191}
\definecolor{dkpink}{rgb}{1.2,0,1.6}
\definecolor{iden}{HTML}{0332FF}
\definecolor{comment}{HTML}{B12122}

\usepackage{etoolbox}

\newlist{conditions}{enumerate}{1}
\setlist[conditions]{
  label=\arabic*.,
  ref=\arabic*
}
\crefname{conditionsi}{condition}{conditions}
\Crefname{conditionsi}{Condition}{Conditions}

\newlist{properties}{enumerate}{1}
\setlist[properties]{
  label=\arabic*.,
  ref=\arabic*
}
\crefname{propertiesi}{property}{properties}
\Crefname{propertiesi}{Property}{Properties}

\newlist{components}{enumerate}{1}
\setlist[components]{
  label=\arabic*.,
  ref=\arabic*
}
\crefname{componentsi}{component}{components}
\Crefname{componentsi}{components}{Components}

\newtheorem{theorem}{Theorem}
\newtheorem{lemma}[theorem]{Lemma}
\newtheorem{proposition}[theorem]{Proposition}
\newtheorem{corollary}[theorem]{Corollary}
\newtheorem{remark}{Remark}

\newcommand{\reason}[1]{\tag*{#1}}
\newcommand{\reasonnext}[1]{\\[-.5ex] & \tag*{\text{#1}}}

\title{Formally Certifying the Vertex Set of a Polyhedron Faster than Informal Enumeration}

\author{Xavier Allamigeon}
\orcid{0000-0002-0258-8018}
\affiliation{\institution{Inria}\institution{CMAP, CNRS, {\'E}cole polytechnique, Institut Polytechnique de Paris}\country{France}}
\email{xavier.allamigeon@inria.fr}
\author{Yazid Id-Sahra}
\affiliation{\institution{\'Ecole Normale Supérieure de Lyon}
\country{France}}
\email{yazid.id-sahra@ens-lyon.fr}
\author{Pierre-Yves Strub}
\orcid{0000-0002-8196-7875}
\affiliation{\institution{PQShield}\country{France}}
\email{pierre-yves.strub@pqshield.com}

\begin{abstract}
The computation of the vertices of a polyhedron described by a system of linear inequalities is a central problem in polyhedral computation. It is a fundamental step in the conversion between H-representations, by linear inequalities, and V-representations, by vertices and extreme rays. This operation plays an important role both in the study of polyhedra and their combinatorics in mathematics and in applications to software and system verification.

We present a certificate-based approach for formally verifying the computation of the vertices of a polyhedron. Given an informally computed list of vertices, our method allows to certify in the proof assistant \Rocq{} that the list is complete, or even exact. The cornerstone of the method is a new completeness criterion based on an abstract simplicial complex that generalizes a triangulation of the normal fan of the polyhedron. A significant advantage over previous approaches is that the usually expensive numerical computations are essentially reduced to membership tests to the polyhedron, while the other steps are cheap combinatorial tests.

We implement the certification method and prove its correctness in the proof assistant \Rocq{}. We experiment with it on a variety of polyhedra, including Birkhoff polytopes, cross-polytopes, cubes, permutahedra, hypersimplices, and high-dimensional polytopes involved in the disproof of the Hirsch conjecture. Our experiments show that certification with the \Rocq{}-to-OCaml extracted checker is typically 1.5x to over 5x faster than vertex enumeration by the state-of-the-art informal C implementation \lrs{} of the reverse search method. 
\end{abstract}

\begin{document}
\maketitle

\section{Introduction}\label{sec:introduction}

Minkowski--Weyl Theorem states that convex polyhedra, \ie, sets of $\R^d$ defined by finitely many linear inequalities, are precisely the convex hulls of finitely many points and rays (\ie, halflines). This statement establishes the existence of a double representation for polyhedra, by means of a system of linear inequalities (H-representation) or by a finite collection of points and rays (V-representation). Some operations over polyhedra, such as intersections of two polyhedra, or testing the membership of a point, can be straightforwardly done over H-representations. Other operations, such as projections, image under linear maps, or convex unions, are immediate over V-representations. This makes the computation of one representation from another a prominent problem in polyhedral computations, which appears in many applications. As an example, in applications to static analysis by abstract interpretation~\cite{CousotHalbwachs1978,BagnaraHillZaffanella2008}, neural network verification~\cite{TranEtAl2020NNV,MullerEtAl2022PRIMA}, or dynamical system analysis~\cite{Guglielmi2017}, it is generally acknowledged that the H- vs V-conversion is one of the main bottleneck steps; see~\eg{}~\cite{SinghPuschelVechev2017,MarechalMonniauxPerin2017}.

In this setting, the computation of the vertices of a polyhedron, or even more specifically, of a polytope (\ie, a bounded polyhedron) plays a central role. Indeed, the vertex set of a polytope provides the smallest (inclusionwise) V-representation, and H- vs V-conversion problems for general polyhedra can be reduced to it (as we recall in \Cref{sec:conclusion}). For this reason, the computation of the vertex set, commonly referred to as the \emph{vertex enumeration problem}, is one of the most fundamental computational problems over polyhedra. It is also closely connected to long-standing questions in combinatorics and optimization involving graphs of polytopes and their diameters. In more details, the (vertex-edge) graph of a polytope is the graph whose vertices are the vertices of the polytope and whose edges are its one-dimensional faces. The diameter is the maximum  graph distance between any two vertices. Diameters of graphs of polytopes provide lower bounds on the worst-case number of steps performed by the simplex method (as the latter actually traces a path to an optimal vertex in such graphs). The Hirsch conjecture asserted that the diameter of the vertex-edge graph of a $d$-dimensional polytope with $n$ facets is at most $n-d$. Although this precise  bound was disproved by~\citet*{Santos2012}, the central question remains whether polytope diameters admit a polynomial bound in $n$ and $d$, which is known as the polynomial Hirsch conjecture. This provides an additional motivation for the development of formally verified methods computing vertex sets and vertex-edge graphs, especially in the context of an increasing use of computer-assisted and AI-assisted methods in the study of open problems.

The state-of-the-art approaches to the vertex enumeration problem are the \emph{double description method}~\cite{MotzkinRaiffaThompsonThrall1953} and the more recent \emph{reverse-search method}~\cite{AvisFukuda1992}. The double description method takes a system of $n$ linear inequalities, and incrementally computes (minimal) V-representations of the polyhedra defined by the first $k$ inequalities for $k = 1, \dots, n$. Despite being very simple to implement, the major drawback of this approach is that the  intermediate polyhedra may have very large (\eg, exponentially large) V-representations compared to that of the final polyhedron. The reverse-search approach is based on the simplex method. It relies on the fact that every vertex corresponds to at least one basis in the sense of the simplex method. It then enumerates all such bases, by reversing the path followed by the simplex method. In non-degenerate situations, \ie, when every vertex corresponds to exactly one simplex basis, the reverse-search method is linear in the number of vertices (up to polynomial factor in the dimension $d$ of the ambient space and the number $n$ of defining inequalities). This applies in particular to simple polytopes, which form a ``dense'' class in the sense that any polytope can be made simple by an arbitrarily small generic perturbation of its defining inequalities. Another advantage of the reverse-search method is that it can be implemented in a memoryless way, and can be parallelized~\cite{AvisJordan2018}. The double description and the reverse-search methods have proofread and highly optimized implementations, respectively the C libraries \cdd{} and \lrs, which are interfaced to various programming languages (\eg, Rust, Julia, Python) and are widely used in mathematical software (\eg, SageMath, Polymake or OSCAR). 

\paragraph{Contribution} We present a certification method for checking that an informally computed list of vertices contains or is equal to the vertex set of a polytope. In addition, it can certify that an informally computed graph is the vertex-edge graph of the polytope. The method takes as input a description of the polytope by inequalities as well as a certificate that is built from the output of the reverse-search method C library \lrs{}. 

The main novelty of our method is the way the completeness of the informal list of vertices is verified. It makes use of the normal fan of the polytope, which is a collection of cones called normal cones, each cone corresponding to a face of the polytope. A basic fact is that the list is complete if and only if the normal cones of the candidate vertices cover the whole space $\R^d$. In order to check this covering property, we incorporate in the certificate an abstract simplicial complex, \ie, a collection of subsets of $[n] \coloneqq \{1, \dots, n\}$ (where $n$ is the number of defining inequalities), which define elementary (conical) tiles contained in these normal cones. These tiles act as a relaxation of a triangulation of the normal cones, in the sense that they are allowed to overlap and may have arbitrary dimension. We establish an \emph{odd-covering criterion} (\Cref{th:odd-covering}) which guarantees that the tiles cover $\R^d$ under a purely combinatorial condition on the abstract simplicial complex and the existence of a (generic) vector covered by an odd number of tiles. As a result, most of the certificate can be checked combinatorially, while the numerical computations are reduced to the verification that candidate vertices belong to the polytope (a rather unavoidable step), and that a given vector is oddly covered. 

We implement the certificate checker in the proof assistant \Rocq{} and formally prove its correctness up to two ingredients: the odd-covering criterion (\Cref{th:odd-covering}) and the correctness statement for vertex-edge graph certification (\Cref{th:graph,th:graph-completeness}). Both involve topological and homological arguments that are not yet available in the \Rocq{} mathematical libraries we use in our development. We therefore provide self-contained pen-and-paper proofs of these results in the present paper and leave their formalization for future work. The correctness of the checker is proved by using proof-oriented types, and relies on the libraries \MathComp{} and \CoqPolyhedra{} for the formalization of the mathematical notions used (fields, vector spaces, polyhedra, faces, etc). This ``high-level'' implementation of the checker is then proved to be equivalent to a ``low-level'' one. The latter involves types made for efficient computations, such as persistent arrays, native binary integers, and the multiprecision arithmetics provided by the \Rocq{} library \C$bignums$. 

We experiment the resulting checker on several classical families of polytopes, namely Birkhoff polytopes, cross-polytopes, cubes, (dual) cyclic polytopes, permutahedra and hypersimplices. Cross-polytopes put aside, each class exhibits an exponentially (or subexponentially) growing number of vertices as the dimension scales up, which make them relevant test cases. We also consider the $20$- and $23$-dimensional polytopes provided as counterexamples to the Hirsch conjecture by~\cite*{Matschke2015}.

We report on two execution modes of the checker. In the first one, the checker runs inside \Rocq{} using its bytecode virtual machine through the tactics \C$vm_compute$, as in the previous work of~\citet*{AllamigeonCanuStrub2023} on the graph certification problem. On their benchmark instances, our method improves the graph certification time by a factor ranging from 15 to 1000. To give an order of magnitude, we are now able to verify each of the two counterexamples to the Hirsch conjecture, including the violation of the Hirsch bound, in resp.\ ~30 and ~70 seconds on a standard laptop, compared to more than an hour in their work. In the second execution mode, the checker is obtained through \Rocq{} extraction mechanism to OCaml code, which we compile to native code. We directly compare its performance to that of \lrs{}. We find that the vertex set certification time is typically 1.5x to over 5x (and up to 50x) faster than the vertex enumeration time. This demonstrates that our certification method could be used on top of \lrs{} in order to provide a formal guarantee at an acceptable overhead.

\paragraph{Related Work}

The closest work to ours is that of~\citet*{AllamigeonCanuStrub2023}, who provide a method for certifying the vertex-edge graph of a polytope and apply it to the formal verification of counterexamples to the Hirsch conjecture. Their method builds on the enumeration of lexicographically feasible bases (lex-feasible bases for short) performed by the reverse-search implementation \lrs{}. As explained in \Cref{subsec:completeness}, we also use lex-feasible bases to construct our certificates, as a post-processing step applied to the output of \lrs{}.
The performance gap between the two approaches can be explained from what is actually certified. As discussed in \Cref{subsec:completeness}, the lex-feasible bases correspond to a particular triangulation of the normal cones called the lexicographic triangulation. The method of~\cite{AllamigeonCanuStrub2023} implicitly amounts to certifying that this specific triangulation is correct. In our implementation, we use the same lexicographic triangulation to generate the abstract simplicial complex contained in the certificate, but the checker does not verify that the complex represents the lexicographic triangulation. Instead, it only checks the weaker properties of the odd-covering criterion to certify the vertex-set and the graph correctness. This saves a substantial amount of numerical computations, most notably matrix-rank tests, as explained in \Cref{sec:odd-covering}. It also results in significantly smaller certificates, as reported in \Cref{subsec:experiments}.

Out of the aforementioned work, we are not aware of any formally proven methods for computing or certifying vertex sets or vertex-edge graphs of polyhedra. To our knowledge, the majority of prior works that use formally verified polyhedral computations avoid dealing directly with H- vs V-representation problems, and reduce the verification of polyhedral operations to solving basic linear inequality problems such as infeasibility or logical implication problems. This is the approach followed by the Verified Polyhedra Library~\cite{BoulmeEtAl2018}, where computational polyhedral operators are implemented in OCaml, and the verification of Farkas certificates (nonnegative combinations of linear inequalities) underlying to the correctness of the polyhedral computations done are discharged to \Rocq{}. Similarly, the formally verified neural network prover Rocq-NN-Roll~\cite{AleksandrovJackischVoellinger2026} relies on solving (in)feasibility problems over linear inequality systems. In practice, rather than exploiting a certification method, they directly do a formal proof of the Fourier--Motzkin elimination method. 

Beyond the context of formal proof, the complexity of the certification problems related to H-/V-representations of polyhedra is an active topic where several questions remain unsolved. In more details, given a polytope $P$ described by some linear inequalities, and finite set $V \subset \R^d$ of points, the complexity of deciding if $V = \vertices(P)$, or even $V \supset \vertices(P)$, is known to be in $\coNP$. However, the membership to $\NP$ or, equivalently, the existence of a polynomial time certification method for the vertex-set exactness or completeness, is still an open question. We refer to~\cite{KaibelPfetsch2003}, in which a broader discussion on the complexity status of problems related to H- vs V-representations can also be found. To our knowledge,~\citet{JoswigZiegler2004} is one of the few works to address vertex-set certification from an algorithmic perspective. Their completeness criterion is stated as the nontriviality of the $(d-1)$th homology group over $\F_2$ (the field of order $2$) of a crosscut complex generated by the sets of vertices incident to each facet. They propose to check this criterion by computing ranks of boundary matrices (of exponential size in the worst case). We are not aware of any implementation (even informal) of this approach. Our odd-covering criterion is also homological in nature, and we provide in \Cref{remark:homology} an interpretation of the conditions of the criterion in terms of homology over $\F_2$.

\paragraph{Outline}

The preliminary \Cref{sec:normal-fan} introduces the key notions underlying to this work, such as polyhedra and their faces, normal cones and normal fans. \Cref{sec:odd-covering} deals with the presentation and the proof of the odd-covering criterion. In \Cref{sec:certificate}, we present the  certification method for vertex sets  and vertex-edge graphs of polytopes (\Cref{subsec:vertex-completeness} and \Cref{subsec:vertex-and-graph}). 
We also show how a certificate can be systematically built from the output of the library \lrs{}, thus establishing the completeness of the certification method (\Cref{subsec:completeness}). The formal proof of the method and its experimental evaluation are covered in \Cref{sec:experiments}. Proofs and extra statements are provided in appendix.

\paragraph{Source Code of the Formalization}

The \Rocq{} formalization is available in the repository \url{https://github.com/Coq-Polyhedra/homology-checker}.
 The \C$README.md$ file describes how to compile it from dependencies, run the checker, and reproduce the experiments reported in \Cref{sec:experiments} (a \C$Dockerfile$ replaying the whole set of benchmarks is also provided in the directory \C$benchmarks$). Certificates can be built from the Rust generator provided in the repository \url{https://github.com/Coq-Polyhedra/lrs-postprocess}. The generator works from the output of a patched version of \C$lrslib$ available in \url{https://github.com/Coq-Polyhedra/lrslib}.

\paragraph{Use of AI} 

Generative AI tools were used to assist the implementation, primarily in two specific contexts. First, they assisted in developing the Rust certificate generator, which lies outside the trusted computing base. Second, they were used to develop the refinement proofs bridging the low- and high-level certificate checkers (see \Cref{subsec:refinement}), whose correctness is then checked by the proof assistant. AI-assisted tools were also used for minor supporting tasks, notably writing benchmark scripts or documentation, and producing the TikZ code of the pictures in the paper.

\section{Polyhedra and Normal Fans}\label{sec:normal-fan}

For a general account on the notions and properties introduced here, we refer the reader to the textbooks~\cite{ziegler,Schrijver86,DeLoeraRambauSantos2010}. We denote by $\scalar{y}{z} \coloneqq \sum_{i = 1}^d y_i z_i$ the Euclidean scalar product of two vectors $y, z \in \R^d$.

A \emph{polyhedron} of $\R^d$ is a set formed by the points $x \in \R^d$ satisfying a finite system of linear (affine) inequalities $\scalar{a_1}{x} \leq b_1$, \dots, $\scalar{a_n}{x} \leq b_n$, where $n \geq 0$, and $a_i \in \R^d$, $b_i \in \R$ for all $i \in [n]$. 

Minkowski--Weyl Theorem entails that bounded polyhedra in $\R^d$ are precisely the sets of the form $\conv(V)$ for $V \subset \R^d$ being a finite set, where $\conv(V)$ stands for the set of convex combinations $\sum_{v \in V} \alpha_v v$ of the points in $V$ ($\alpha_v \in \R_{\geq 0}$ for all $v \in V$, $\sum_{v \in V} \alpha_v = 1$). Sets of the form $\conv(V)$ for $V$ finite are called \emph{polytopes}. This is why we use the more standard term ``polytope'' to designate a bounded polyhedron. However, we confirm that the polytopes dealt with in the paper will be all described by inequalities as input (while the goal will be to compute or certify a corresponding set $V$). 

We also extensively use polyhedral cones. A \emph{polyhedral cone} is a set of $x \in \R^d$ satisfying a finite system of linear inequalities $\scalar{f_i}{x} \geq 0$ where $f_i \in \R^d$, $i \in [n]$ and $n \geq 0$. Polyhedral cones are special cases of polyhedra. (A polyhedral cone is nothing but a polyhedron closed under the scaling $x \mapsto \beta x$ by all scalar $\beta > 0$). Minkowski--Weyl Theorem states that polyhedral cones are precisely the sets of the form $\cone(W)$ for $W \subset \R^d$ being a finite set, where $\cone(W)$ stands for the set of nonnegative combinations $\sum_{w \in W} \beta_w w$ of the vectors in $W$ ($\beta_w \geq 0$ for all $w \in W$). 

An important notion is that of the \emph{dimension} of a polyhedron. The latter is defined as the dimension of the \emph{affine hull} of the polyhedron, \ie, of the smallest (inclusionwise) affine subspace of $\R^d$ containing the polyhedron. For example, a polyhedron reduced to a single point has dimension $0$, a line segment has dimension $1$, \etc. In the case of a polytope under the form $\conv(V)$ where $V = \{v^{(0)}, \dots, v^{(p)}\}$ for $p \geq 0$, the dimension is given by the rank of family of vectors $v^{(i)} - v^{(0)}$ where $i = 1, \dots, p$. Analogously, the dimension of a polyhedral cone under the form $\cone(W)$ for $W$ finite is given by the rank of the family $W$. We denote by $\dim P$ the dimension of a polyhedron $P$. 

\begin{figure*}[t]
\begin{tikzpicture}[
  scale=.55,
  convex/.style={draw=none,fill=lightgray,fill opacity=.7},
  convexborder/.style={very thick,black!80},
  normal/.style={ultra thick,shorten <=5pt,-{Stealth[length=7pt,width=6pt]}},
  edgenormal/.style={thick,gray!90,shorten <=5pt,-{Stealth[length=7pt,width=4pt]}},
  fanray/.style={thick,gray!90},
  lbl/.style={font=\footnotesize}
]
  \clip (-1.1,-1) rectangle (25.5,9);
  \draw[gray!70,->] (-1,0) -- (10.5,0) node[above left] {$x_1$};
  \draw[gray!70,->] (0,-1) -- (0,8.5) node[below right] {$x_2$};

  \coordinate (v1) at (3,1);
  \coordinate (v2) at (7,2);
  \coordinate (v3) at (8,5);
  \coordinate (v4) at (3,7);
  \coordinate (v5) at (1,3);

  \fill[convex] (v1) -- (v2) -- (v3) -- (v4) -- (v5) -- cycle;
  \draw[convexborder] (v1) -- (v2) -- (v3) -- (v4) -- (v5) -- cycle;

  \draw[normal,vertexone]   (v1) -- ++(-.30,-1.75);
  \draw[normal,vertextwo]   (v2) -- ++(1.40,-1.75);
  \draw[normal,vertexthree] (v3) -- ++(1.45,1.16);
  \draw[normal,vertexfour]  (v4) -- ++(-.52,1.70);
  \draw[normal,vertexfive]  (v5) -- ++(-1.75,0);

  \draw[edgenormal] (5,1.5) -- ++(.32,-1.28)
    node[lbl,below right] {$a_1$};
  \draw[edgenormal] (7.5,3.5) -- ++(1.20,-.40)
    node[lbl,right] {$a_2$};
  \draw[edgenormal] (5.5,6) -- ++(.48,1.20)
    node[lbl,above] {$a_3$};
  \draw[edgenormal] (2,5) -- ++(-1.15,.58)
    node[lbl,above left] {$a_4$};
  \draw[edgenormal] (2,2) -- ++(-.85,-.85)
    node[lbl,below left] {$a_5$};

  \fill[vertexone]   (v1) circle (4pt);
  \fill[vertextwo]   (v2) circle (4pt);
  \fill[vertexthree] (v3) circle (4pt);
  \fill[vertexfour]  (v4) circle (4pt);
  \fill[vertexfive]  (v5) circle (4pt);

  \begin{scope}[shift={(18,4)}]
    \coordinate (o) at (0,0);

    \begin{scope}
      \clip (o) -- (225:5) arc[start angle=225,end angle=284.04,radius=5] -- cycle;
      \shade[inner color=vertexone!38,middle color=vertexone!38,outer color=white] (o) circle (4.15);
    \end{scope}
    \begin{scope}
      \clip (o) -- (284.04:5) arc[start angle=284.04,end angle=341.57,radius=5] -- cycle;
      \shade[inner color=vertextwo!38,middle color=vertextwo!38,outer color=white] (o) circle (4.15);
    \end{scope}
    \begin{scope}
      \clip (o) -- (341.57:5) arc[start angle=341.57,end angle=428.20,radius=5] -- cycle;
      \shade[inner color=vertexthree!38,middle color=vertexthree!38,outer color=white] (o) circle (4.15);
    \end{scope}
    \begin{scope}
      \clip (o) -- (68.20:5) arc[start angle=68.20,end angle=153.43,radius=5] -- cycle;
      \shade[inner color=vertexfour!38,middle color=vertexfour!38,outer color=white] (o) circle (4.15);
    \end{scope}
    \begin{scope}
      \clip (o) -- (153.43:5) arc[start angle=153.43,end angle=225,radius=5] -- cycle;
      \shade[inner color=vertexfive!38,middle color=vertexfive!38,outer color=white] (o) circle (4.15);
    \end{scope}

    \draw[fanray] (o) -- (225:3.85);
    \draw[fanray] (o) -- (284.04:3.85);
    \draw[fanray] (o) -- (341.57:3.85);
    \draw[fanray] (o) -- (68.20:3.85);
    \draw[fanray] (o) -- (153.43:3.85);
    \node[lbl,below right] at (284.04:4.02) {$a_1$};
    \node[lbl,right]       at (341.57:4.02) {$a_2$};
    \node[lbl,above right] at (68.20:4.02) {$a_3$};
    \node[lbl,above left]  at (153.43:4.02) {$a_4$};
    \node[lbl,below left]  at (225:4.02) {$a_5$};

    \draw[normal,vertexone]   (o) -- ++(-.30,-1.75);
    \draw[normal,vertextwo]   (o) -- ++(1.40,-1.75);
    \draw[normal,vertexthree] (o) -- ++(1.45,1.16);
    \draw[normal,vertexfour]  (o) -- ++(-.52,1.70);
    \draw[normal,vertexfive]  (o) -- ++(-1.75,0);
    \fill[black!80] (o) circle (3pt);
    \node[lbl,below right] at (o) {$0$};
  \end{scope}
\end{tikzpicture}

\caption{A two-dimensional polytope (left) and its normal fan (right). Each colored vector defines a linear function uniquely maximized at the associated vertex; the same vector is shown inside the corresponding normal cone. The generating rays in the normal fan are labeled by the corresponding outer normal vector $a_i$ of the polytope.}\label{fig:normal-fan}
\end{figure*}
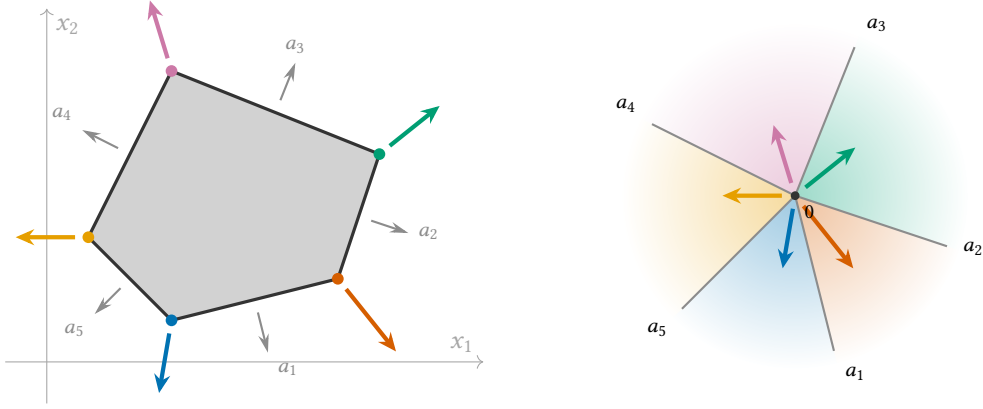

Let $P \subset \R^d$ be a polytope. We recall that a \emph{vertex} of~$P$ is a point $v \in P$ arising as the unique maximizer of some linear function $x \mapsto \scalar{c}{x}$ over $P$ (where $c \in \R^d$), \ie,  
\[
\text{for all} \; x \in P \, , \; x \neq v \implies \scalar{c}{x} < \scalar{c}{v} \, .
\]
We refer to \Cref{fig:normal-fan} for an example. We denote by $\vertices(P)$ the set of vertices of $P$. As discussed in \Cref{sec:introduction}, the vertex set is the minimal (inclusionwise) set $V \subset \R^n$ such that $P = \conv(V)$. This explains the term ``V-representation'' in which the letter ``V'' refers to vertices (while the H in ``H-representation'' refers to halfspaces, \ie, sets defined by a single affine inequality $\scalar{a}{x} \leq b$).

More generally, a set of points maximizing some linear function $x \mapsto \scalar{c}{x}$ over $P$ is called a (nonempty) \emph{face} of $P$. Faces are polytopes as well. The vertices of $P$ correspond to the faces of dimension $0$. The \emph{edges} are the faces of dimension $1$. The \emph{vertex-edge graph} of $P$ is the combinatorial graph $G(P)$ whose vertices correspond to the vertices of $P$, and there is an edge between two vertices $v, w$ in $G(P)$ if $v$ and $w$ belong to a common edge in the polytope $P$. The \emph{facets} of $P$ are the faces of dimension $\dim P - 1$. 

Following the definition of a vertex, a natural approach to check if some informally computed set $V \subset \R^d$ contains all vertices of $P$ is to check that for every linear function with a unique point as maximizer, the latter belongs to $V$. While apparently this would need to examine infinitely many vectors $c \in \R^n$, the notion of normal fan yields a finite discretization making this approach applicable.

In more details, the \emph{normal cone} $N(x)$ of a point $x \in P$ is defined as the set of vectors $c \in \R^d$ such that $x$ is maximizing the function $x \mapsto \scalar{c}{x}$ over $P$, \ie,
\[
N(x) \coloneqq \{ c \in \R^d \colon \forall y \in P, \scalar{c}{x} \geq \scalar{c}{y} \} \, .
\]
Now, suppose that $P$ is defined by the inequalities $\scalar{a_i}{x} \leq b_i$ where $a_i \in \R^d$ and $b_i \in \R$ for all $i \in [n]$. Let $x \in P$. It is well known (see~\Cref{prop:normal-cone-def} in \Cref{app:normal-fan}) that the normal cone $N(x)$ is the cone generated by the vectors $a_i$ of the inequalities \emph{active} at $x$, \ie, satisfying $\scalar{a_i}{x} = b_i$. Formally, we have
\begin{equation}
N(x) = \cone{\{a_i\}_{i \in I(x)}}  \, . \label{eq:normal-cone-characterization}
\end{equation}
where $I(x) \coloneqq \{ i \in [n] \colon \scalar{a_i}{x} = b_i \}$. We refer again to \Cref{fig:normal-fan} for an illustration, where we observe that the normal cone of every vertex is indeed generated by the normal vectors $a_i$ of two corresponding active inequalities.

The \emph{normal fan} is the collection of normal cones $N(x)$ for $x \in P$. Observe that the characterization given in~\eqref{eq:normal-cone-characterization} in terms of a subset $I(\cdot) \subset [n]$ shows that the normal fan consists of finitely many cones. The term \emph{fan} refers to the fact that the faces of every cone in the collection also belong to the collection, and that the intersection of any two cones is a common face of the two. 

We are now ready to state the proposition that will be at the core of the certification of vertex-set completeness.
\begin{proposition}\label{prop:vertex-completeness}
Suppose that $P$ is a nonempty polyhedron, and let $V \subset P$. The following two properties are equivalent:
\begin{itemize}
\item $P$ is a polytope and $V$ contains all vertices of $P$;
\item $\bigcup_{v \in V} N(v) = \R^d$.
\end{itemize}
\end{proposition}

\section{Odd-Covering Criterion}\label{sec:odd-covering}

\subsection{Overview: the Role of Triangulations}\label{subsec:overview}

The criterion provided by \Cref{prop:vertex-completeness} turns the vertex-set completeness problem $V \supset \vertices(P)$ into a covering of the whole space by the normal cones $N(v)$ of the points $v \in V$. Equivalently, as normal cones are invariant by scaling by a positive scalar, we can visualize this problem as a covering problem of the $(d-1)$-dimensional unit sphere $S^{d-1} \subset \R^d$ by the spherical polytopes $N(v) \cap S^{d-1}$ for $v \in V$, as illustrated in the middle of \Cref{fig:cross-polytope} with the normal fan of the three-dimensional cross-polytope. (We recall that the $d$-dimensional cross-polytope is the set of the points $x \in \R^d$ satisfying $\sum_{i = 1}^d |x_i| \leq 1$, or, equivalently, the $2^d$ linear inequalities $\sum_{i = 1}^d \eps_i x_i \leq 1$ where $\eps \in \{+1,-1\}^d$.)

The approach to such covering problems is to check that the whole boundary of every normal cone --- here, the sets $N(v)$ --- is fully covered by the boundary of other normal cones. However, in general, the boundary of a normal cone, as described by its facets, may be arbitrarily complex. While the normal cones are naturally given by some V-representation (see~\eqref{eq:normal-cone-characterization}), enumerating their facets falls in the same family of problems as the vertex enumeration (see \Cref{sec:conclusion}). Triangulations, \ie, decomposition in triangles in dimension $2$, or simplices in higher dimension, are of the main techniques to overcome such issues. For cones, such triangulations are stated in terms of simplicial cones. A ($k$-dimensional) \emph{simplicial cone} is a cone of the form $\cone \{w_1, \dots, w_k\}$ generated by $k$ linearly independent vectors $w_1, \dots, w_k \in \R^d$. Checking if a collection of simplicial cones forms a complete triangulation of $\R^d$ (\ie, covering all $\R^d$) is made easier by the fact that every $d$-dimensional simplicial cone has precisely $d$ facets, and each facet can be found out by solving a $(d-1) \times (d-1)$ linear system. These are the foundations of the criteria developed in~\cite[Corollary~4.5.20]{DeLoeraRambauSantos2010} or~\cite[Proposition~3.1]{Pilaud_2025}. They can be informally described as follows:
\begin{theorem}[Informal statement]\label{th:covering-triangulation}
A finite collection $\Ccal$ of $d$-dimensional simplicial cones in $\R^d$
is a complete triangulation of $\R^d$ if and only if
\begin{conditions}
\item\label{cond:icop} each facet of a cone in $\Ccal$ is a facet of precisely two cones in $\Ccal$, one of each side of the facet;
\item\label{cond:ipp} there is a vector $z^\star \in \R^d$ contained in the interior
of one cone of $\Ccal$ and which does not belong to any other cone in $\Ccal$.
\end{conditions}
\end{theorem}
The right-hand side of \Cref{fig:cross-polytope} illustrates this criterion by depicting a decomposition of the normal fan into simplicial cones which, when intersected with $S^2$, provides a triangulation of the sphere. As an illustration of \Cref{cond:icop}, the darker spherical edge depicted in the foreground is contained in exactly two triangles, one on each side.

\begin{figure*}[t]

\definecolor{vplusone}{RGB}{230,159,0}
\definecolor{vplustwo}{RGB}{204,121,167}
\definecolor{vplusthree}{RGB}{86,180,233}
\definecolor{vminusone}{RGB}{0,114,178}
\definecolor{vminustwo}{RGB}{213,94,0}
\definecolor{vminusthree}{RGB}{0,158,115}

\pgfplotsset{
  colormap={vplusonemap}{color(0cm)=(vplusone);color(1cm)=(vplusone)},
  colormap={vplustwomap}{color(0cm)=(vplustwo);color(1cm)=(vplustwo)},
  colormap={vplusthreemap}{color(0cm)=(vplusthree);color(1cm)=(vplusthree)},
  colormap={vminusonemap}{color(0cm)=(vminusone);color(1cm)=(vminusone)},
  colormap={vminustwomap}{color(0cm)=(vminustwo);color(1cm)=(vminustwo)},
  colormap={vminusthreemap}{color(0cm)=(vminusthree);color(1cm)=(vminusthree)},
  fan boundary back/.style={no marks,draw=black!32,line width=.30pt,
    densely dotted,opacity=.72,samples y=0},
  fan boundary front/.style={no marks,draw=black!42,line width=.35pt,
    samples y=0},
  fan diagonal back/.style={no marks,draw=black!28,line width=.25pt,
    densely dotted,opacity=.68,samples y=0},
  fan diagonal front/.style={no marks,draw=black!38,line width=.30pt,
    samples y=0},
  fan featured edge/.style={no marks,draw=black!68,line width=1.00pt,
    samples y=0}
}

\def\viewaz{118}
\def\viewel{18}

\def\triU{(-1+2*x)}
\def\triV{(-1+2*x*y)}
\def\triNorm{sqrt(1+(\triU)^2+(\triV)^2)}
\newcommand{\TriangularShade}[5]{%
  \addplot3[surf,shader=interp,faceted color=,
    colormap name=#1,draw=none,opacity=#2,
    domain=0:1,y domain=0:1,samples=8]
    ({#3},{#4},{#5});%
}

\begin{tikzpicture}[scale=1]
  \begin{axis}[
    at={(0,0)},anchor=center,
    hide axis,view={\viewaz}{\viewel},axis equal image,
    width=8.6cm,height=8.6cm,
    xmin=-1.18,xmax=1.18,ymin=-1.18,ymax=1.18,zmin=-1.18,zmax=1.18,
    scale mode=scale uniformly,
    clip=false
  ]
    \addplot3[gray!43,line width=.45pt]
      coordinates {(-1.2,0,0) (1,0,0)};
    \addplot3[gray!43,line width=.45pt,-stealth]
      coordinates {(0,-1.2,0) (0,1.2,0)};
    \addplot3[gray!43,line width=.45pt,-stealth]
      coordinates {(0,0,-1.2) (0,0,1.2)};
    \node[gray!62,font=\footnotesize,inner sep=1pt]
      at (axis cs:1.4,0,0) {$x_1$};
    \node[gray!62,font=\footnotesize,inner sep=1pt]
      at (axis cs:0,1.3,0) {$x_2$};
    \node[gray!62,font=\footnotesize,inner sep=1pt]
      at (axis cs:0,0,1.3) {$x_3$};

    \foreach \sy in {-1,1}{
      \addplot3[black!30,line width=.8pt,dashed]
        coordinates {(-1,0,0) (0,\sy,0)};
    }
    \foreach \sz in {-1,1}{
      \addplot3[black!30,line width=.8pt,dashed]
        coordinates {(-1,0,0) (0,0,\sz)};
    }

    \foreach \sy in {-1,1}{
      \foreach \sz in {-1,1}{
        \addplot3[fill=black!12,draw=none,fill opacity=.16]
          coordinates {(1,0,0) (0,\sy,0) (0,0,\sz) (1,0,0)};
      }
    }

    \foreach \sy in {-1,1}{
      \addplot3[black!62,line width=1.05pt]
        coordinates {(1,0,0) (0,\sy,0)};
    }
    \foreach \sz in {-1,1}{
      \addplot3[black!62,line width=1.05pt]
        coordinates {(1,0,0) (0,0,\sz)};
    }
    \foreach \sy in {-1,1}{
      \foreach \sz in {-1,1}{
        \addplot3[black!62,line width=1.05pt]
          coordinates {(0,\sy,0) (0,0,\sz)};
      }
    }

    \addplot3[gray!43,line width=.45pt,-stealth]
      coordinates {(1,0,0) (1.28,0,0)};

    \addplot3[only marks,mark=*,mark size=2.65pt,vplusone,
      mark options={fill opacity=.88,draw=white,draw opacity=.8,line width=.35pt}]
      coordinates {(1,0,0)};
    \addplot3[only marks,mark=*,mark size=2.65pt,vplustwo,
      mark options={fill opacity=.88,draw=white,draw opacity=.8,line width=.35pt}]
      coordinates {(0,1,0)};
    \addplot3[only marks,mark=*,mark size=2.65pt,vplusthree,
      mark options={fill opacity=.88,draw=white,draw opacity=.8,line width=.35pt}]
      coordinates {(0,0,1)};
    \addplot3[only marks,mark=*,mark size=2.65pt,vminusone,
      mark options={fill opacity=.88,draw=white,draw opacity=.8,line width=.35pt}]
      coordinates {(-1,0,0)};
    \addplot3[only marks,mark=*,mark size=2.65pt,vminustwo,
      mark options={fill opacity=.88,draw=white,draw opacity=.8,line width=.35pt}]
      coordinates {(0,-1,0)};
    \addplot3[only marks,mark=*,mark size=2.65pt,vminusthree,
      mark options={fill opacity=.88,draw=white,draw opacity=.8,line width=.35pt}]
      coordinates {(0,0,-1)};
  \end{axis}

  \newcommand{\NormalFan}[2]{%
  \begin{axis}[
    at={(#1,0)},anchor=center,
    hide axis,view={\viewaz}{\viewel},axis equal image,
    width=7.2cm,height=7.2cm,
    xmin=-1.12,xmax=1.12,ymin=-1.12,ymax=1.12,zmin=-1.12,zmax=1.12,
    scale mode=scale uniformly
  ]
    \pgfmathsetmacro{\camx}{sin(\viewaz)*cos(\viewel)}
    \pgfmathsetmacro{\camy}{-cos(\viewaz)*cos(\viewel)}
    \pgfmathsetmacro{\camz}{sin(\viewel)}

    \foreach \sy in {-1,1}{
      \addplot3[fan boundary back,domain=-1:1,samples=25,variable=\u]
        ({-1/sqrt(2+\u^2)},{\sy/sqrt(2+\u^2)},{\u/sqrt(2+\u^2)});
    }
    \foreach \sz in {-1,1}{
      \addplot3[fan boundary back,domain=-1:1,samples=25,variable=\u]
        ({-1/sqrt(2+\u^2)},{\u/sqrt(2+\u^2)},{\sz/sqrt(2+\u^2)});
    }

    \foreach \sy/\sz in {1/1,1/-1,-1/1,-1/-1}{
      \pgfmathsetmacro{\tsil}{-(\camy*\sy+\camz*\sz)/\camx}
      \addplot3[fan boundary back,domain=-1:\tsil,samples=18,variable=\u]
        ({\u/sqrt(2+\u^2)},{\sy/sqrt(2+\u^2)},{\sz/sqrt(2+\u^2)});
    }

    \ifnum#2=1\relax
    \addplot3[fan diagonal back,domain=-1:1,samples=25,variable=\u]
      ({-1/sqrt(1+2*\u^2)},{\u/sqrt(1+2*\u^2)},{\u/sqrt(1+2*\u^2)});
    \foreach \sy in {-1,1}{
      \pgfmathsetmacro{\tdiag}{-\camy*\sy/(\camx+\camz)}
      \addplot3[fan diagonal back,domain=-1:\tdiag,samples=18,variable=\u]
        ({\u/sqrt(1+2*\u^2)},{\sy/sqrt(1+2*\u^2)},{\u/sqrt(1+2*\u^2)});
    }
    \foreach \sz in {-1,1}{
      \pgfmathsetmacro{\tdiag}{-\camz*\sz/(\camx+\camy)}
      \addplot3[fan diagonal back,domain=-1:\tdiag,samples=18,variable=\u]
        ({\u/sqrt(1+2*\u^2)},{\u/sqrt(1+2*\u^2)},{\sz/sqrt(1+2*\u^2)});
    }
    \fi

    \addplot3[surf,shader=interp,faceted color=,
      colormap name=vminusonemap,draw=none,opacity=.08,
      domain=-1:1,y domain=-1:1,samples=8]
      ({-1/sqrt(1+x^2+y^2)},{x/sqrt(1+x^2+y^2)},{y/sqrt(1+x^2+y^2)});
    \ifnum#2=1\relax
      \TriangularShade{vminusonemap}{.08}
        {-1/\triNorm}{\triU/\triNorm}{\triV/\triNorm}
    \fi
    \addplot3[surf,shader=interp,faceted color=,
      colormap name=vminustwomap,draw=none,opacity=.08,
      domain=-1:1,y domain=-1:1,samples=8]
      ({x/sqrt(1+x^2+y^2)},{-1/sqrt(1+x^2+y^2)},{y/sqrt(1+x^2+y^2)});
    \ifnum#2=1\relax
      \TriangularShade{vminustwomap}{.08}
        {\triU/\triNorm}{-1/\triNorm}{\triV/\triNorm}
    \fi
    \addplot3[surf,shader=interp,faceted color=,
      colormap name=vminusthreemap,draw=none,opacity=.08,
      domain=-1:1,y domain=-1:1,samples=8]
      ({x/sqrt(1+x^2+y^2)},{y/sqrt(1+x^2+y^2)},{-1/sqrt(1+x^2+y^2)});
    \ifnum#2=1\relax
      \TriangularShade{vminusthreemap}{.08}
        {\triU/\triNorm}{\triV/\triNorm}{-1/\triNorm}
    \fi

    \addplot3[surf,shader=interp,faceted color=,
      colormap name=vplustwomap,draw=none,opacity=.24,
      domain=-1:1,y domain=-1:1,samples=8]
      ({x/sqrt(1+x^2+y^2)},{1/sqrt(1+x^2+y^2)},{y/sqrt(1+x^2+y^2)});
    \ifnum#2=1\relax
      \TriangularShade{vplustwomap}{.30}
        {\triU/\triNorm}{1/\triNorm}{\triV/\triNorm}
    \fi
    \addplot3[surf,shader=interp,faceted color=,
      colormap name=vplusthreemap,draw=none,opacity=.24,
      domain=-1:1,y domain=-1:1,samples=8]
      ({x/sqrt(1+x^2+y^2)},{y/sqrt(1+x^2+y^2)},{1/sqrt(1+x^2+y^2)});
    \ifnum#2=1\relax
      \TriangularShade{vplusthreemap}{.14}
        {\triU/\triNorm}{\triV/\triNorm}{1/\triNorm}
    \fi
    \addplot3[surf,shader=interp,faceted color=,
      colormap name=vplusonemap,draw=none,opacity=.24,
      domain=-1:1,y domain=-1:1,samples=8]
      ({1/sqrt(1+x^2+y^2)},{x/sqrt(1+x^2+y^2)},{y/sqrt(1+x^2+y^2)});
    \ifnum#2=1\relax
      \TriangularShade{vplusonemap}{.30}
        {1/\triNorm}{\triU/\triNorm}{\triV/\triNorm}
    \fi

    \ifnum#2=1\relax
    \addplot3[fan diagonal front,domain=-1:1,samples=25,variable=\u]
      ({1/sqrt(1+2*\u^2)},{\u/sqrt(1+2*\u^2)},{\u/sqrt(1+2*\u^2)});
    \foreach \sy in {-1,1}{
      \pgfmathsetmacro{\tdiag}{-\camy*\sy/(\camx+\camz)}
      \addplot3[fan diagonal front,domain=\tdiag:1,samples=18,variable=\u]
        ({\u/sqrt(1+2*\u^2)},{\sy/sqrt(1+2*\u^2)},{\u/sqrt(1+2*\u^2)});
    }
    \foreach \sz in {-1,1}{
      \pgfmathsetmacro{\tdiag}{-\camz*\sz/(\camx+\camy)}
      \addplot3[fan diagonal front,domain=\tdiag:1,samples=18,variable=\u]
        ({\u/sqrt(1+2*\u^2)},{\u/sqrt(1+2*\u^2)},{\sz/sqrt(1+2*\u^2)});
    }
    \fi

    \foreach \sy in {-1,1}{
      \addplot3[fan boundary front,domain=-1:1,samples=25,variable=\u]
        ({1/sqrt(2+\u^2)},{\sy/sqrt(2+\u^2)},{\u/sqrt(2+\u^2)});
    }
    \foreach \sz in {-1,1}{
      \addplot3[fan boundary front,domain=-1:1,samples=25,variable=\u]
        ({1/sqrt(2+\u^2)},{\u/sqrt(2+\u^2)},{\sz/sqrt(2+\u^2)});
    }

    \foreach \sy/\sz in {1/1,1/-1,-1/1,-1/-1}{
      \pgfmathsetmacro{\tsil}{-(\camy*\sy+\camz*\sz)/\camx}
      \addplot3[fan boundary front,domain=\tsil:1,samples=18,variable=\u]
        ({\u/sqrt(2+\u^2)},{\sy/sqrt(2+\u^2)},{\sz/sqrt(2+\u^2)});
    }

    \ifnum#2=1\relax
    \addplot3[fan featured edge,domain=-1:1,samples=35,variable=\u]
      ({1/sqrt(2+\u^2)},{1/sqrt(2+\u^2)},{\u/sqrt(2+\u^2)});
    \fi
  \end{axis}%
  }%

  \NormalFan{7.2cm}{0}
  \NormalFan{13.2cm}{1}
\end{tikzpicture}

\caption{The three-dimensional cross-polytope (left), its normal fan intersected with the sphere $S^2$ (middle), and a triangulation of the normal fan (right).}\label{fig:cross-polytope}
\end{figure*}
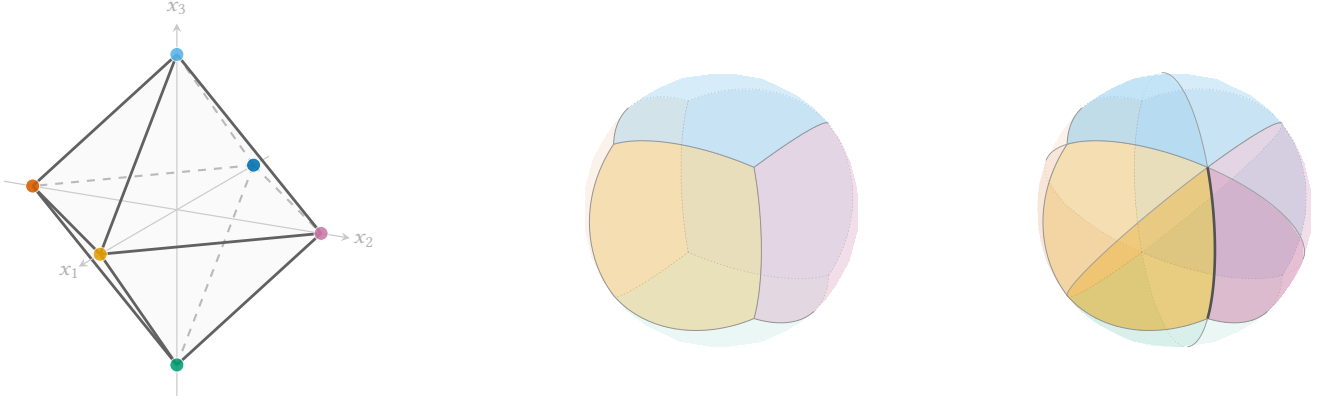

The drawback of the criterion of \Cref{th:covering-triangulation} is that \Cref{cond:icop} requires a significant amount of numerical computations. Indeed, we need to compute every facet of each normal cone, which is done by solving a $(d-1) \times (d-1)$ linear system. Hence, the number of such systems to be solved is proportional to the number of simplicial cones in the triangulation, which may be very large, as explained in \Cref{subsec:experiments}.

As discussed in \Cref{sec:introduction}, we propose a significant relaxation of the conditions of \Cref{th:covering-triangulation} that allows to replace a large part of the numerical burden by combinatorial tests, thus reducing the impact of a large size triangulation. To this aim, we rely on abstract simplicial complexes, which will encode decompositions in which the elementary cells are not required to be $d$-dimensional anymore, nor to intersect only on their boundaries. 

\subsection{The Statement}

An \emph{abstract simplicial complex} over a finite set $E$ is a collection $\Ksf$ of finitely many sets contained in $E$ and closed under subsets. The sets constituting $\Ksf$ are called the \emph{(abstract) simplices}. The closedness property of simplices under subsets is analogous to that exhibited by simplicial cones and their faces: given a $k$-dimensional simplicial cone generated by the vectors $w_1, \dots, w_k$, its faces are precisely the cones generated by the subsets of $\{w_1, \dots, w_k\}$. The \emph{dimension} of an abstract simplex $\sigma$ is defined as $\# \sigma - 1$, where $\# \sigma$ is the cardinality of the set $\sigma$. The \emph{dimension} of an abstract simplicial complex $\Ksf$, denoted $\dim \Ksf$ is defined as the maximum dimension of its simplices. In the sequel, we restrict to the case where $\Ksf$ is \emph{pure}, \ie, all maximal (inclusionwise) simplices have the same dimension (equal to $\dim \Ksf)$. Under this assumption, a simplex of dimension $\dim \Ksf$ is called a \emph{facet} of $\Ksf$, and a simplex of dimension $(\dim \Ksf) -1$ is called a \emph{ridge}. 

In our setting, we use pure abstract simplicial complexes over the ground set $E = [n]$ of dimension $d-1$ (\ie, the dimension of the sphere $S^{d-1}$ rather than that of $\R^d$, due to homogeneity). Thus, facets are subsets of $[n]$ with cardinality~$d$, and ridges have cardinality $d-1$. Such an abstract complex $\Ksf$ can be realized geometrically into a collection of cones by mapping every element $i \in [n]$ to the vector $a_i$, and, by extension, every simplex $\sigma \in \Ksf$ to the cone $a_\sigma \coloneqq \cone{\{a_i\}_{i \in \sigma}}$. The dimension of the cone $a_\sigma$ is upper bounded by $\# \sigma$. Equality holds if and only if the $a_i$ for $i \in \sigma$ are linearly independent, \ie, $a_\sigma$ is a simplicial cone. Simpliciality will not be required in general in the conditions of the odd-covering criterion below. We also point out a subtle terminology discrepancy between polyhedral fans and abstract simplicial complexes that may be a source of confusion: in the case the cone $a_\sigma$ is a $d$-dimensional simplicial cone, the simplex $\sigma$ is a facet of the abstract simplicial complex, while the facets of $a_\sigma$ correspond to cones $a_\tau$ where $\tau$ is a ridge. We recommend the reader to think of the ridges as the analogous of the facets of cones discussed in \Cref{subsec:overview}.

\begin{example}\label{example:asc}
\Cref{fig:odd-covering} can be used to illustrate a case where $d = 3$, $n = 6$. The abstract simplicial complex $\Ksf$ is given by the facets $\{1, 2, i\}$ for $i = 3, \dots, 6$ and their subsets. All four facets share the ridge $\{1,2\}$. Cones are again represented by the intersection with the sphere $S^2$. As expected from the generality of the setting, the cones $a_\sigma$ for $\sigma$ facets may intersect over their interiors (\eg, $a_{\{1,2,4\}}$, $a_{\{1,2,5\}}$ and $a_{\{1,2,6\}}$). To keep the illustration readable, we do not depict the case in which some of these cones are not $3$-dimensional.
\end{example}

We start by stating a covering result for $d$-generic vectors (\Cref{th:odd-covering}). A vector $x \in \R^d$ is $k$-\emph{generic} if $x$ does not belong to any cone $a_\sigma$ satisfying $\dim a_\sigma < k$. 
The set of $d$-generic vectors is open dense in $\R^d$, as shown in \Cref{lemma:open-dense} in \Cref{app:odd-covering}. Moreover, as the boundary of cones $a_\sigma$ is made of cones of dimension less than $d$, a $d$-generic vector can only belong to the interior of ($d$-dimensional) cones $a_\sigma$.

The main result of this section is the following:
\begin{theorem}\label{th:odd-covering}
Let $\Ksf$ be a pure $(d-1)$-dimensional abstract simplicial complex over the set $[n]$, and a collection of nonzero vectors $(a_i)_{i \in [n]} \subset \R^d$. Suppose that
\begin{conditions}
\item\label{cond:ridge} each ridge of $\Ksf$ is contained in an even number of facets;
\item\label{cond:root} there is a $d$-generic vector $z^\star \in \R^d$ such that $z^\star \in a_\sigma$ for an odd number of facets $\sigma$ of $\Ksf$;
\item\label{cond:pointed} for all facet $\sigma$ of $\Ksf$, $a_\sigma$ is a pointed cone, \ie, $a_\sigma \cap (-a_\sigma) = \{0\}$.
\end{conditions}
Then, every $d$-generic vector $u \in \R^d$ satisfies $u \in a_\sigma$ for an odd number of facets $\sigma$ of $\Ksf$. 
\end{theorem}

As outlined above, the conditions of \Cref{th:odd-covering} generalize those of \Cref{th:covering-triangulation}. (As we shall see in \Cref{sec:certificate}, \Cref{cond:pointed} of \Cref{th:odd-covering} is a technical condition that will be automatically fulfilled thanks to a global condition on the dimension of the polytope $P$ to which we will apply the certification method.) A first difference is that the conditions of \Cref{th:odd-covering} only involve the parity of incidence numbers, rather than an exact counting. A much more significant difference is that \Cref{cond:ridge} in \Cref{th:odd-covering} is only stated in terms of the abstract simplicial complex, \ie, in terms of combinatorial conditions over the incidence of abstract ridges and facets, rather than resorting to a geometric condition like \Cref{cond:ridge} in \Cref{th:covering-triangulation} involving the solutions of linear systems. 

We now give an overview of the proof provided in \Cref{app:odd-covering}. We consider a $d$-generic vector $u$, and a $C^1$ path $\gamma \colon [0,1] \to \R^d$ from a vector $z^\star$ satisfying \Cref{cond:root} to the vector $u$, \ie, $\gamma(0) = z^\star$ and $\gamma(1) = u$. We require this path to be \emph{generic}, in the sense that it only intersects cones $a_\sigma$ ($\sigma \in \Ksf$) of dimension $d$ or $d-1$, and the intersection with every $(d-1)$-dimensional such cone is transversal (\ie, the tangent vector of the path at the intersection point is not parallel to the affine hull of the $(d-1)$-dimensional cone). Under such assumptions, $\gamma$ intersects $(d-1)$-dimensional cones only at finitely many values of $t \in [0,1]$ (\Cref{lemma:finite-breakpoints}). At these values, the vector $\gamma(t)$ is $(d-1)$-generic. Out of these values, \ie, on open intervals of $[0,1]$, the vector $\gamma(t)$ is $d$-generic. We show (\Cref{lemma:linear-path}) that such a generic path can be built from a linear path $t \mapsto z^\star + t (u - z^\star)$ up to slightly perturbing the vector $z^\star$ (which is possible since \Cref{cond:root} still holds after such a perturbation).

Then, we keep track of the parity of the cardinality of the set $S(t) \coloneqq \{ \sigma \; \text{facet of} \; \Ksf \colon \gamma(t) \in a_\sigma \}$ by continuously increasing $t$ from~$0$ to $1$. On every open interval where $\gamma(t)$ remains $d$-generic, the path does not cross the boundary of any cone (as such boundary consists of union of lower-dimensional cones), and thus $S(t)$ remains constant. The only changes on $S(t)$ may appear at a value $t^\circ$ where $\gamma(t^\circ)$ is $(d-1)$-generic but not $d$-generic. At such a breakpoint, $\gamma(t^\circ)$ lies on the facets of some cones that the path is leaving or entering. Such facets are of the form $a_\tau$ where $\tau$ is a ridge of $\Ksf$. The fact that every such ridge appears in an even number of facets $\sigma$ of $\Ksf$ allows to show that the symmetric difference $S(t^-) \symdiff S(t^+)$ has even cardinality. We refer to \Cref{fig:odd-covering} for an illustration (which must be now understood as the illustration of only a part of a larger abstract simplicial complex than the one described in \Cref{example:asc}). Proceeding in this way up to $t = 1$, we obtain that $S(1)$ has odd cardinality like $S(0)$.

\begin{figure}[t]
\definecolor{facetone}{RGB}{230,159,0}
\definecolor{facettwo}{RGB}{86,180,233}
\definecolor{facetthree}{RGB}{204,121,167}
\definecolor{facetfour}{RGB}{0,158,115}

\tikzset{
  facet side/.style={line width=.70pt,opacity=.82},
  facet vertex/.style={circle,inner sep=1.45pt,draw=white,
    line width=.25pt},
  ridge/.style={draw=black!72,line width=1.15pt},
  path/.style={draw=black!58,line width=.75pt,-{Stealth[length=4.5pt]}},
  lbl/.style={font=\footnotesize,text=black!72},
  vertex lbl/.style={font=\footnotesize,text=black!72,inner sep=1pt}
}

\begin{tikzpicture}[x=1cm,y=1cm,scale=.8]
  \coordinate (A) at (0,-2.15);
  \coordinate (B) at (0, 2.15);
  \coordinate (C1) at (-2.30, 0.00);
  \coordinate (C2) at ( 2.12,-1.10);
  \coordinate (C3) at ( 2.50, 0.00);
  \coordinate (C4) at ( 2.08, 1.08);
  \coordinate (crossing) at (-.09,0);

  \path[fill=facettwo,fill opacity=.20]
    (A) -- (C2) -- (B)
    .. controls (-.12,.75) and (-.12,-.75) .. (A) -- cycle;
  \path[fill=facetthree,fill opacity=.20]
    (A) -- (C3) -- (B)
    .. controls (-.12,.75) and (-.12,-.75) .. (A) -- cycle;
  \path[fill=facetfour,fill opacity=.20]
    (A) -- (C4) -- (B)
    .. controls (-.12,.75) and (-.12,-.75) .. (A) -- cycle;

  \path[fill=facetone,fill opacity=.25]
    (A) -- (C1) -- (B)
    .. controls (-.12,.75) and (-.12,-.75) .. (A) -- cycle;

  \draw[facet side,draw=facettwo!70!black] (A) -- (C2) -- (B);
  \draw[facet side,draw=facetthree!70!black] (A) -- (C3) -- (B);
  \draw[facet side,draw=facetfour!70!black] (A) -- (C4) -- (B);
  \draw[facet side,draw=facetone!70!black] (A) -- (C1) -- (B);

  \draw[ridge]
    (A) .. controls (-.12,-.75) and (-.12,.75) .. (B);

  \node[facet vertex,fill=facetone] at (C1) {};
  \node[facet vertex,fill=facettwo] at (C2) {};
  \node[facet vertex,fill=facetthree] at (C3) {};
  \node[facet vertex,fill=facetfour] at (C4) {};
  \fill[black!62] (A) circle (1.35pt);
  \fill[black!62] (B) circle (1.35pt);

  \node[vertex lbl,below=4pt]       at (A)  {$a_1$};
  \node[vertex lbl,above=4pt]       at (B)  {$a_2$};
  \node[vertex lbl,left=4pt]        at (C1) {$a_3$};
  \node[vertex lbl,below right=3pt] at (C2) {$a_4$};
  \node[vertex lbl,right=4pt]       at (C3) {$a_5$};
  \node[vertex lbl,above right=3pt] at (C4) {$a_6$};

  \draw[path] (-1.18,0) -- (1.22,0);
  \fill[black!70] (-.72,0) circle (1.25pt);
  \fill[black!70] ( .72,0) circle (1.25pt);
  \filldraw[draw=black!72,fill=white,line width=.45pt]
    (crossing) circle (1.45pt);

  \node[lbl,below left={2pt and -3pt}] at (-.6,0) {$\gamma(t^-)$};
  \node[lbl,above right={1pt and .5pt}] at (crossing) {$\gamma(t^\circ)$};
  \node[lbl,below right={1pt and -3pt}] at (.6,0) {$\gamma(t^+)$};
\end{tikzpicture}

\caption{An illustration of the odd-covering criterion and its proof, in the case $d = 3$. The only cone $a_\tau$ (for $\tau$ being a ridge) met by the path $\gamma$ at $t = t^\circ$ is $\tau = \{1,2\}$. The number of cones $a_\sigma$ for $\sigma$ facet of $\Ksf$ containing the path point $\gamma(t)$ varies from $1$ at $t = t^-$ to $3$ at $t = t^+$.}\label{fig:odd-covering}
\end{figure}
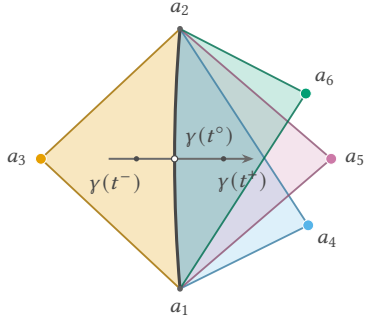

\begin{remark}\label{remark:homology}
We can interpret the conditions of \Cref{th:odd-covering} in homological terms as follows. (We refer to~\cite{Hatcher,Munkres} for definitions and notation.) \Cref{cond:ridge} amounts to the fact that the $\F_2$-chain 
\[
c \coloneqq \sum_{\sigma \, \text{facet of}\, \Ksf} \sigma \in C_{d-1}(\Ksf; \F_2) 
\]
is a cycle. Recall that we can also think of every $\sigma \in \Ksf$ as a singular simplex, \ie, as the map sending $t \in \triangle_{\dim \sigma}$ to $\sum_{i \in \sigma} t_i e_i \in \R^n$ (where $e_i$ is the $i$th unit vector, and $\triangle_k \coloneqq \{ x \in \R_{\geq 0}^{k+1} \colon \sum_{i = 1}^k x_i = 1 \}$). By extension, $c$ can be seen a singular chain in $C_{d-1}(|\Ksf|; \F_2)$ where $|\Ksf| \coloneqq \bigcup_{\sigma \in \Ksf} \conv \{e_i \colon i \in \sigma\}$. In this setting, $c$ is still a cycle.

\Cref{cond:pointed} ensures that the map $f$ associating the point $\sum_{i \in \sigma} t_i e_i \in |\Ksf|$ to $\frac{z}{\|z\|} \in S^{d-1}$ where $z = \sum_{i \in \sigma} t_i a_i$, is well-defined (\ie, $z \neq 0$). As a consequence, the push-forward chain $f_\sharp(c)$ is a cycle on $S^{d-1}$. 

The interpretation of \Cref{cond:root} involves local homology groups $H_{d-1}(S^{d-1}, S^{d-1} \setminus \{y\}; \F_2)$ on $S^{d-1}$ ($y \in S^{d-1}$); see \cite[\S35]{Munkres} for a definition. We recall that the elements of $H_{d-1}(S^{d-1}, S^{d-1} \setminus \{y\}; F_2)$ are represented by local cycles, \ie, chains of $C_{d-1}(S^{d-1}; \F_2)$ whose boundary consists of singular simplices whose image does not contain $y$. We denote by $[\cdot]_y$ the homology class of such cycles. 

Then, \Cref{cond:root} states that $[f_\sharp(c)]_{y^\star}$ is not trivial. Indeed, let $y \in S^{d-1}$ be a $d$-generic point. As shown in \Cref{lemma:local-homology-simplex}, for all facet $\sigma$ of $\Ksf$, we have $[f_\sharp(\sigma)]_{y} \neq 0$ if and only if $y \in a_\sigma$. Therefore, $[f_\sharp(c)]_y \neq 0$ if and only if $y \in S^{d-1}$ belongs to an odd number of cones $a_\sigma$ for $\sigma$ facet of $\Ksf$. 

By \Cref{lemma:local-homology-isomorphic}, the inclusion maps from $H_{d-1}(S^{d-1}; \F_2)$ to the local homology groups are isomorphisms. Hence, $[f_\sharp(c)]_{y^\star} \neq 0$ implies $[f_\sharp(c)]_y \neq 0$ for all $y \in S^{d-1}$. We deduce that every $d$-generic point $y \in S^{d-1}$ belongs to an odd number of cones $a_\sigma$ for $\sigma$ facet of $\Ksf$. This yields an alternative (though not self-contained) proof of \Cref{th:odd-covering}.
\end{remark}

\begin{remark}\label{remark:pure}
The assumption that $\Ksf$ is pure was made only to make the presentation easier. It is not needed in the proof given in \Cref{app:odd-covering} (nor in the alternative argument of \Cref{remark:homology}).
\end{remark}

The following result about the covering of $\R^d$ (and not only $d$-generic vectors) is an immediate corollary. It is a consequence of the density of $d$-generic vectors in $\R^d$ (\Cref{lemma:open-dense}) and the closedness of the cones~$a_\sigma$.
\begin{corollary}\label{cor:covering}
Under the assumptions of \Cref{th:odd-covering}, the cones $a_\sigma$ where $\sigma$ ranges over the facets of $\Ksf$ cover the space $\R^d$.
\end{corollary}

\section{The Certification Method}\label{sec:certificate}

Throughout this section, we suppose that $P$ is a nonempty polyhedron defined by $n$ inequalities $\scalar{a_i}{x} \leq b_i$ where $a_i \in \R^d$, $a_i \neq 0$, and $b_i \in \R$ for $i \in [n]$. Without loss of generality, we assume that the vectors $a_i$ are nonzero.\footnote{The case where $a_i = 0$ corresponds either to a tautological inequality when $b_i \geq 0$, in which case it can be eliminated from the system, or to a invalid inequality when $b_i < 0$, in which case the polyhedron is trivially empty.} (This condition is used to apply \Cref{th:odd-covering}.)

\subsection{Vertex-Set Completeness}\label{subsec:vertex-completeness}

Let $V$ be a finite set of points of $\R^d$. The certificate that we use to prove $V \supset \vertices(P)$ consists of the following components:
\begin{components}
\item for each $v \in V$, a set $I_v \subset [n]$;
\item a collection $F$ of subsets of $[n]$ with cardinality $d$, and a mapping $\phi \colon F \to V$;
\item a (undirected) graph $G$ with vertex set $F$;
\item a special subset $\sigma^\star \in F$ associated to $v^\star \coloneqq \phi(\sigma^\star)$, a collection of vectors $f_1, \dots, f_d \in \R^d$, and for all $\sigma \in \phi^{-1}(v^\star) \setminus \{\sigma^\star\}$, a vector $s_\sigma \in \R^d$.
\end{components}

The verification of the certificate consists of the following steps:
\begin{conditions}[label=T\arabic*., ref=T\arabic*]
\item\label{cond:feasibility} for all $v \in V$, $\scalar{a_i}{v} = b_i$ if $i \in I_v$, and $\scalar{a_i}{v} < b_i$ otherwise;
\item\label{cond:facet_subset} for all $\sigma \in F$, $\sigma$ is a subset of $I_{\phi(\sigma)}$;
\item\label{cond:ridge_count} every $\sigma \in F$ has degree at most $d$ in $G$, and, for each $i \in \sigma$, there is a neighbor $\sigma'$ satisfying $\sigma' \neq \sigma$ and $\sigma \setminus \{i\} \subset \sigma'$; 
\item\label{cond:rank} setting $\sigma^\star = (i_1, \dots, i_d)$, $\scalar{a_{i_k}}{f_j} > 0$ if $j = k$, and  $\scalar{a_{i_k}}{f_j} = 0$ otherwise;
\item\label{cond:sep} for all $\sigma \in \phi^{-1}(v^\star) \setminus \{\sigma^\star\}$, $s_\sigma \in \R_{\geq 0}^d$, $s_\sigma \neq 0$, and for every $i \in \sigma$, $\scalar[\big]{a_i}{\sum_j (s_\sigma)_j f_j} \leq 0$.
\end{conditions}

In order to prove that the vertex-set completeness property is implied by  \Cref{cond:feasibility,cond:facet_subset,cond:ridge_count,cond:rank,cond:sep}, we use the abstract simplicial complex $\Ksf$ generated by the elements in $F$, \ie, consisting of the subsets of $\sigma$ for all $\sigma \in F$. As every $\sigma \in F$ has cardinality $d$, the complex $\Ksf$ is pure and $(d-1)$-dimensional, and $F$ is precisely the set of its facets. Our goal is to show that we can apply \Cref{th:odd-covering} to $\Ksf$ and deduce that the cones $a_\sigma$ cover the whole space. Given $v \in V$, the normal cone $N(v)$ is intended to be covered by the cones $a_\sigma$ for the elements $\sigma \in F$ mapped to $v$ by $\phi$. In this case, we will be able to conclude that $V \supset \vertices(P)$ thanks to \Cref{prop:vertex-completeness}.

We start with two lemmas showing that $V \subset P$, and relating each cone $a_\sigma$ with the normal cone $N(v)$ of the associated point $v = \phi(\sigma)$.
\begin{lemma}\label{lemma:feasibility}
Under \Cref{cond:feasibility}, we have $v \in P$, $I_v = I(v)$, and $N(v) = \cone{\{a_i\}_{i \in I_v}}$ for all $v \in V$.
\end{lemma}

\begin{lemma}\label{lemma:facet-subset-normal-cone}
Under \Cref{cond:feasibility,cond:facet_subset}, $\phi(\sigma)$ belongs to $P$ and $a_\sigma \subset N(\phi(\sigma))$ for every $\sigma \in F$.
\end{lemma}

We now study the properties of the graph $G$ induced by \Cref{cond:ridge_count}. As shown in  \Cref{lemma:neighbors,lemma:simple}, the condition first implies that $G$ is a simple graph (\ie, without loops nor multiple edges) in which every $\sigma \in F$ has precisely $d$ neighbors $\sigma^{(i)} \in F$, one per element $i \in \sigma$. These neighbors satisfy $\sigma \setminus \{i\} \subset \sigma^{(i)}$ for all $i \in \sigma$. The set $\sigma \setminus \{i\}$ is a ridge of $\Ksf$, and every ridge of $\Ksf$ arises in this way (for some $\sigma \in \Ksf$). Moreover, it is contained in the two facets $\sigma$ and $\sigma^{(i)}$. 

Given a ridge $\tau$ of $\Ksf$, we show that every facet $\sigma$ containing $\tau$ comes by pairs of adjacent facets $(\sigma, \sigma^{(i)})$ in the graph $G$ (here, $i$ is the unique element of $\sigma$ such that $\tau = \sigma \setminus \{i\}$). This is the idea underlying to the proof of the next statement, which establishes that \Cref{cond:ridge_count} entails \Cref{cond:ridge} of \Cref{th:odd-covering}. The complete proof is provided in \Cref{app:certificate}.

\begin{proposition}\label{lemma:ridge-even}
Under \Cref{cond:ridge_count}, every ridge of $\Ksf$ appears in an even number of facets. 
\end{proposition}

Next, we show that \Cref{cond:root} of \Cref{th:odd-covering} holds. To this aim, we use the special facet $\sigma^\star$ involved in \Cref{cond:rank,cond:sep}. The idea is the following. \Cref{cond:rank} implies that the cone $a_{\sigma^\star}$ is a simplicial cone of dimension $d$. In fact, it can be shown that the vectors $f_1, \dots, f_d$ represents (inner) normal vectors to its facets, \ie, the cone $a_{\sigma^\star}$ is the set of vectors $u \in \R^d$ such that $\scalar{f_i}{u} \geq 0$ for all $i \in [d]$. We then pick an arbitrary vector $z^\star$, for instance $z^\star \coloneqq \sum_{i \in \sigma^\star} a_i$, in the interior of $a_{\sigma^\star}$. By the property of the normal fan (\Cref{cor:normal-cone-basic-prop}), $z^\star$ cannot lie in any normal cone $N(w)$ for $w \neq v^\star$, hence, not in any cone $a_\sigma$ satisfying $\phi(\sigma) \neq v^\star$ by \Cref{lemma:facet-subset-normal-cone}. Thanks to \Cref{cond:sep}, we show that $a_{\sigma^\star}$ is the only cone in $N(v^\star)$ containing $z^\star$ by exhibiting a hyperplane $\Big\{ z \in \R^n \colon \scalar[\big]{\sum_j (s_\sigma)_j f_j}{z} = 0 \Big\}$ separating $z^\star$ from any cone $a_\sigma$ satisfying $\phi(\sigma) = v^\star$, $\sigma \neq \sigma^\star$. In this way, we arrive at the following statement.

\begin{proposition}\label{prop:root}
Under \Cref{cond:feasibility,cond:facet_subset,cond:rank,cond:sep}, the vector $z^\star = \sum_{i \in \sigma^\star} a_i$ is $d$-generic, and the only cone $a_\sigma$ ($\sigma \in \Ksf$) containing it is $a_{\sigma^\star}$. 
\end{proposition}

In order to fullfill \Cref{cond:pointed} of \Cref{th:odd-covering}, it suffices to assume that $P$ is a full-dimensional. Indeed, under this assumption, every normal cone is pointed (see \Cref{lemma:normal-cone-pointed}), and so are the cones $a_\sigma$ owing the inclusion $a_\sigma \subset N(\phi(\sigma))$ (\Cref{lemma:facet-subset-normal-cone}). We are now ready to show the correctness of our certification method.

\begin{theorem}\label{th:vertex-completeness}
Let $P$ be a full-dimensional polyhedron.

If \Cref{cond:feasibility,cond:facet_subset,cond:ridge_count,cond:rank,cond:sep} are satisfied, $P$ is a polytope and the set $V$ contains all its vertices.
\end{theorem}

\begin{remark}\label{remark:complexity}
We briefly discuss the complexity of the certification method. All basic operations involved, namely computing scalar products in \Cref{cond:feasibility,cond:rank,cond:sep} and subset tests in \Cref{cond:facet_subset,cond:ridge_count}, can be done in polynomial time in~$d$ and~$n$. Since $\phi^{-1}(v^\star) \subset F$, it follows that the verification of \Cref{cond:feasibility,cond:facet_subset,cond:ridge_count,cond:rank,cond:sep} can be carried out in time $O(\poly(d, n) (\# V + \# F))$ where $\poly(d, n)$ denotes some polynomial factor in $d$ and $n$. As we shall see in \Cref{subsec:completeness}, in practice we take $V = \vertices(P)$ and $F$ a triangulation of the normal fan of $P$. In this setting, we have $\# F \geq \# V$ (\ie, there is at least one facet per vertex), so that the total complexity simply writes as $O(\poly(d,n) \#F)$. If, moreover, the polytope $P$ is simple, all normal cones are simplicial, and $F$ consists of exactly one facet per vertex. In this case, the certification method becomes polynomial in $d$, $n$ and $\# (\vertices(P))$. 

As explained in \Cref{subsec:experiments} in the case of cross-polytopes, there may be exponentially many more simplices in $F$ than vertices. In this situation, as remarked in~\cite{JoswigZiegler2004}, a possibly better approach consists in applying the certification method to the polar polytope~\cite[\S~2.3]{ziegler}. In more details, up to replacing the full-dimensional polytope $P \subset \R^d$ by $P - x^\circ$ where $x^\circ$ is an interior point, the inequalities defining $P$ can be put under the form $\scalar{\bar a_i}{x} \leq 1$ where $\bar a_i \in \R^d$, $i \in [n]$. We claim that it suffices to apply \Cref{th:vertex-completeness} to the polyhedron $Q \coloneqq \{ z \in \R^n \colon \forall v \in V \, , \; \scalar{v}{z} \leq 1 \}$ (which is full-dimensional as $0$ is an interior point), and certify that $Q$ is a polytope and the points $(\bar a_i)_{i \in [n]}$ contain its vertices. Indeed, the latter two properties are equivalent to $Q \subset \conv{\{\bar a_i\}_{i \in [n]}}$. The polar of the polytope $Q$ is $\conv{V}$, and the polar of $\conv{\{\bar a_i\}_{i \in [n]}}$ is the polyhedron $P = \{x \in \R^d \colon \forall i \in [n], \, \scalar{\bar a_i}{x} \leq 1 \}$. Since passing to  polars reverses set inclusion, we obtain that $P \subset \conv{V}$, which is equivalent to $V \supset \vertices(P)$ and $P$ being a polytope. This proves the claim. The polar-based approach is especially useful when $P$ is a simplicial polytope (like in the case of cross-polytopes). Indeed, under this assumption, the (candidate) polar $Q$ is a simple polytope. Following the discussion in the previous paragraph, the complexity of the certification method then becomes $O(\poly(d, \#(\vertices(P))) n)$.
\end{remark}

\subsection{Vertex Set and Vertex-Edge Graph}\label{subsec:vertex-and-graph}

In order to certify that $V = \vertices(P)$ from the results of \Cref{subsec:vertex-completeness}, it suffices to verify that every point $v \in V$ is a vertex of $P$. The standard way of certifying that a point $v \in P$ is a vertex is to show that there are $d$ linearly independent vectors among the vectors $a_i$ for $i \in I(v)$, \ie, satisfying $\scalar{a_i}{v} = b_i$ (see e.g.~\cite[Chapter~8]{Schrijver86}). However, this would reintroduce an expensive rank-based condition that we managed to avoid in the certification $\vertices(P) \subset V$. Instead, we propose a purely combinatorial approach that uses the following extra component in the certificate
\begin{components}[resume]
\item for each $v \in V$, a sequence $i_1, \dots, i_d \in I_v$ and points $w_1, \dots, w_d \in V$,
\end{components}
for which we verify the condition:
\begin{conditions}[resume,label=T\arabic*., ref=T\arabic*]
\item\label{cond:flag} for all $v \in V$, we have 
\[
\forall k \in [d] \, , \; i_k \notin I_{w_k} \enspace \text{and} \enspace \forall j < k \, , \; i_k \in I_{w_j} \, .
\]
\end{conditions}

\begin{proposition}\label{prop:flag}
Suppose that \Cref{cond:feasibility} holds, and let $v \in V$ satisfying \Cref{cond:flag}. Then the polyhedron~$P$ is full-dimensional and $v$ is a vertex of $P$.
\end{proposition}
While the proof we provide in appendix is purely algebraic, we shall see in \Cref{subsec:completeness} that \Cref{cond:flag} is related to the existence of a \emph{flag} above the point $v$ of $P$, \ie, a chain of faces $F_0 = \{v\} \subset F_1 \subset \dots \subset F_{\dim(P)}$ such that $\dim F_i = i$ for all $i$. The points $w_1, \dots, w_d$ can be built from vertices of the faces $F_i$, and the fact that the chain has length $d$ ensures $\dim P = d$.

The following result is an immediate consequence.
\begin{theorem}\label{cor:vertex-set}
If \Cref{cond:feasibility,cond:facet_subset,cond:ridge_count,cond:rank,cond:sep,cond:flag} are satisfied, $P$ is a full-dimensional polytope and the set $V$ is precisely the set of vertices of $P$.
\end{theorem}

We now describe the certification of the vertex-edge graph $G(P)$ of $P$. The candidate graph is the image $\phi(G)$ of the graph $G$ by the mapping $\phi$, \ie, the non-directed graph with vertex set $\phi(F) \subset V$ and whose edges are the pairs $\{v,w\} \subset \phi(F)$ such that $v \neq w$ and there is an edge $\{\sigma, \sigma'\}$ in $G$ such that $\sigma \in \phi^{-1}(v)$ and $\sigma' \in \phi^{-1}(w)$. 
\begin{theorem}\label{th:graph-completeness}
If \Cref{cond:feasibility,cond:facet_subset,cond:ridge_count,cond:rank,cond:sep,cond:flag} are satisfied, $G(P)$ is a subgraph of $\phi(G)$.
\end{theorem}

\Cref{th:graph-completeness} entails that $G(P)$ and $\phi(G)$ have the same vertex set (since $\vertices(V) \subset \phi(F) \subset V = \vertices(P)$), and the edge set of $G(P)$ is contained in that of $\phi(G)$. We point out that even though it does not certify that $\phi(G)$ is exactly vertex-edge graph of $P$, it shows that the diameter of $G(P)$ is lower bounded by that of $\phi(G)$ (as the two graphs have the same vertices, and the latter have more edges). Consequently, \Cref{th:graph-completeness} can be already used to certify lower bounds on diameters like in the disproof of the Hirsch conjecture.

Finally, we introduce the following condition to ensure that every edge of $\phi(G)$ is an actual edge of $G(P)$. 
\begin{conditions}[resume,label=T\arabic*., ref=T\arabic*]
\item\label{cond:incomparable-edges} for every $v \in V$, the sets $I_v \setminus I_w$ for $w$ adjacent to $v$ in $\phi(G)$ are incomparable. 
\end{conditions}

\begin{theorem}\label{th:graph}
If \Cref{cond:feasibility,cond:facet_subset,cond:ridge_count,cond:rank,cond:sep,cond:flag,cond:incomparable-edges} are satisfied, the graph $\phi(G)$ is precisely the vertex-edge graph $G(P)$ of $P$. 
\end{theorem}

\subsection{Building a Certificate from Lex-Feasible Bases}\label{subsec:completeness}

We show that the previously described certificate can always be built in the case where $P$ is a full-dimensional polytope, proving that our certification method is complete for this class. We recall that, given a general polytope $P$, we can always reduce to the full-dimensional setting by eliminating $d - \dim(P)$ variables in the inequality system.

The construction we propose actually makes use of the enumeration of the so-called lexicographically feasible bases (or lex-feasible bases for short), as described in~\cite{Avis2000} and computed by \lrs{}. Consider the perturbation $\tilde P$ of the polyhedron $P$ defined as $\tilde P \coloneqq \{x \in \R^d \colon \forall i \in [n] \, , \; \scalar{a_i}{x} \leq b_i + \eps^i \}$, where $\eps > 0$. As shown in~\cite{Avis2000}, assuming that $\eps$ is small enough, every vertex of $\tilde P$ activates precisely $d$ inequalities. Moreover, assuming $\eps$ below some threshold values, 
the sets of active inequalities $\{ i \in [n] \colon \scalar{a_i}{\tilde v} = b_i + \eps^i \}$ for $\tilde v$ ranging over the vertices of $\tilde P$ are independent of the choice of $\eps$. These sets are called the \emph{lex-feasible bases}. As proved in~\cite{Avis2000}, every lex-feasible basis $I$ is a feasible basis of $P$, meaning that the linear system $\scalar{a_i}{x} = b_i$ for $i \in I$ has a unique solution, say $v$, and this solution belongs to $P$. Thus, $v$ is a vertex of $P$, and $I \subset I(v)$. In short, every lex-feasible basis corresponds to a vertex of $P$. The following statement extends this relationship to normal fans. We recall that a \emph{simplicial fan} is a fan consisting of simplicial cones.
\begin{proposition}\label{prop:simplicial-refinement}
The normal fan of $\tilde P$ is a simplicial fan which refines the normal fan of $P$, \ie, every cone in the normal fan of $\tilde P$ is contained in some cone of the normal fan of $P$. 
\end{proposition}
It can be shown that the refinement provided by the normal fan $\tilde P$ corresponds to a \emph{lexicographic triangulation} of the normal fan of $P$, which is a special case of a regular triangulation (see~\cite[Chapters~2 and~4]{DeLoeraRambauSantos2010}).

We now build $F$ as the set of lex-feasible bases. We define $\phi$ as the mapping which associates a feasible basis $I$ to the (unique) vertex of $P$ such that $I \subset I(v)$. We take for $G$ the adjacency graph of lex-feasible bases, \ie, there is an edge between two bases $I, I'$ if $I$ and $I'$ differ by only one element. We set $V \coloneqq \vertices(P)$, and $I_v \coloneqq I(v)$ for $v \in V$.

\begin{proposition}\label{prop:t1-to-t3-complete}
The previously described components $V$, $I_v$ ($v \in V$), $F$, $\phi$ and $G$ satisfy \Cref{cond:feasibility,cond:facet_subset,cond:ridge_count}.
\end{proposition}

We are free to choose $\sigma^\star$ as an arbitrary lex-feasible basis, and we set $v^\star \coloneqq \phi(\sigma^\star)$. Let us write $\sigma^\star = (i_1, \dots, i_d)$ where $i_1, \dots, i_d \in [n]$. By definition of a basis, the matrix $\begin{pmatrix} a_{i_1} & \dots & a_{i_d} \end{pmatrix}$ has rank $d$. It then suffices to take $\begin{pmatrix} f_1 & \dots & f_d \end{pmatrix}$ as the transposed inverse matrix. As we already explained in the description of \Cref{prop:root}, the simplicial cone $a_{\sigma^\star}$ is then defined by the $d$ inequalities $\scalar{f_i}{u} \geq 0$. Besides, for each $\sigma \in \phi^{-1}(\sigma^\star)$, $\sigma \neq \sigma^\star$, the existence of a hyperplane $\Big\{ z \in \R^n \colon \scalar[\big]{\sum_j (s_\sigma)_j f_j}{z} = 0 \Big\}$ weakly separating $a_{\sigma^\star}$ and $a_{\sigma}$ is a classic separation result related to Farkas Lemma, that applies because $a_\sigma$ and $a_{\sigma^\star}$ can only meet in lower dimensional cones (as a consequence of the structure of the normal fan of $\tilde P$; see \Cref{cor:normal-cone-basic-prop}). In our case, we can actually build the weights $s_\sigma \in \R^d$ from the vertex $\tilde x^{(\sigma)}$ of $\tilde P$ corresponding to the lex-feasible basis $\sigma$ respectively:
\begin{proposition}\label{prop:t4-t5-complete}
Let $\sigma^\star = (i_1, \dots, i_d)$ be an arbitrary lex-feasible basis, $\begin{pmatrix} f_1 & \dots & f_d \end{pmatrix}$ the transposed inverse matrix of $\begin{pmatrix} a_{i_1} & \dots & a_{i_d} \end{pmatrix}$ and $s_\sigma \coloneqq (b_{i_j} + \eps^{i_j} - \scalar{a_{i_j}}{x^{(\sigma)}})_{j \in [d]}$.

Then, \Cref{cond:rank,cond:sep} are satisfied.
\end{proposition}

The completeness of \Cref{cond:flag} follows from the existence of a flag $F_0 = \{v\} \subset F_1 \subset \dots \subset F_{d-1} \subset F_d = P$ above every vertex $v$ of $P$ (we refer to~\cite{ziegler}). For every $k$, $F_{k-1}$ is a proper face of $F_k$, and so there is a vertex $w_k$ of $F_k$ that does not belong to $F_{k-1}$. The point $w_k$ is a vertex of $P$, thus it belongs to $V$. We take for $i_k$ any index of inequality active over the face $F_{k-1}$ but not at $w_k$, \ie, $\scalar{a_{i_k}}{x} = b_{i_k}$ for all $x \in F_{k-1}$ and $\scalar{a_{i_k}}{w_k} < b_{i_k}$. In this way, $i_k \notin I_{w_k}$. Moreover, for all $k \in [d]$ and $j < k$, we have $i_k \in I_{w_j}$ since $w_j \in F_j \subset F_k$. The same argument shows that $i_1, \dots, i_d \in I_v$. We readily get the following result:
\begin{proposition}\label{prop:flag-complete}
For each $v \in V$, every sequence $i_1, \dots, i_d \in I_v$ and vertices $w_1, \dots, w_d \in V$ built from a flag as above satisfy \Cref{cond:flag}. 
\end{proposition}

We finally deal with \Cref{cond:incomparable-edges}. As proved in~\cite{AllamigeonCanuStrub2023}, the image by $\phi$ of the adjacency graph of lex-feasible bases of $P$ coincides with the vertex-edge graph $G(P)$ of $P$. With this setting, \Cref{cond:incomparable-edges} is then equivalent to the fact that for every vertex $v$ of $P$, the sets  $I(v) \setminus I(w)$ for $w$ adjacent to $v$ in $G(P)$ is adjacent are incomparaoble. Equivalently, the sets $I(v) \cap I(w)$ for $w$ adjacent to $v$ in $G(P)$ are incomparable. Notice that $I(v) \cap I(w)$ is precisely the set of inequalities active over the edge $[v,w]$ of $P$. As these sets of active inequalities uniquely characterizes the faces of $P$ (see~\cite[Chapter~8]{Schrijver86}), that of edges must be incomparable. We deduce that \Cref{cond:incomparable-edges} is satisfied.

\section{Efficient Implementation and Experiments}\label{sec:experiments}

The formalization of the certification method is provided in the project \C$homology-checker$ (see \Cref{sec:introduction} for a link to the repository). The main dependencies are the \Rocq{} libraries \CoqPolyhedra{} (which provides a formalization of the theory of convex polyhedra and their faces~\cite{AllamigeonKatzStrub2022}) and \MathComp{}.\footnote{The \CoqPolyhedra{} library currently depends on versions 1.x of \MathComp{}, which are not available for \Rocq{} 9.x. We have updated \MathComp{} and \CoqPolyhedra{} to support the most recent version 8.x of \Rocq, \ie, Coq~8.20.1. The patched \MathComp{} library and related packages are available at \url{https://github.com/Coq-Polyhedra}.}

The certification method is formalized into two versions of a checker, one which uses ``high-level'' proof-oriented types provided by \CoqPolyhedra{} and \MathComp{}, and another based on ``low-level'' computation-oriented types. The former is used to carry out a correctness proof of the certification method (\Cref{subsec:high-level}), whereas the latter is meant to be executed directly in \Rocq{} or via extraction to OCaml. An equivalence proof between the high-level and the low-level is done; see \Cref{subsec:low-level}. 

\subsection{Correctness Proof of the Certification Method}\label{subsec:high-level}

The high-level version of the checker is implemented in the file \C$src/HighLevelCertificate.v$. The certificate is formalized by means of a record type \C$Certificate$ that follows the description given in \Cref{sec:certificate}. \Cref{cond:feasibility,cond:facet_subset,cond:ridge_count,cond:rank,cond:sep,cond:flag,cond:incomparable-edges} are implemented through functions of type \C$Certificate -> bool$. There are also a number of additional test functions which verify the ``well-formedness'' of the components of the certificate that are implicit in the presentation of \Cref{sec:certificate}: for instance, the fact that the graph $G$ has vertex set $F$, the membership of $\sigma^\star$ to the set $F$, and the membership of the points $w_1, \dots, w_d$ to $V$ for each $v \in V$. 

Some preliminary notions and properties introduced in \Cref{sec:normal-fan}, in particular normal cones, are formalized in the file \C$src/covering_criterion.v$, where \Cref{prop:vertex-completeness} is proved. The odd-covering criterion of \Cref{sec:odd-covering} is stated in the file \C$src/odd_covering_theorem.v$, together with the definition and basic properties of abstract simplicial complexes. It reproduces the statement of \Cref{th:odd-covering} (more precisely, of its generalization, \Cref{cor:covering}) as follows:
\begin{lstlisting}
Definition |*ridgesHaveEvenIncidence*| K :=
  forall R : simplex_n, R \in ridgesOf K 
    -> even #|[set F in facetsOf K | R \subset F]|.§\smallskip§
Definition |*existsSpecialPoint*| K :=
  exists z : 'cV[R]_d, 
       isKGeneric (d.+1) z 
    /\ odd #|[set F in facetsOf K | z \in coneOf F]|.§\smallskip§
Definition |*conesArePointed*| K :=
  forall F : simplex_n, F \in facetsOf K -> polyhedron.pointed (coneOf F).§\smallskip§
Theorem |*odd_covering_theorem*| K :
     ridgesHaveEvenIncidence K
  -> existsSpecialPoint K
  -> conesArePointed K
  -> forall x : 'cV[R]_d, exists F : simplex_n, 
       F \in facetsOf K /\ x \in coneOf F.
\end{lstlisting}
The variable \C$K$ represents the abstract simplicial complex $\Ksf$ (the hypothesis on the fact that it is pure $d$-dimensional here is not reported here). The sets of facets and ridges are respectively provided by the terms \C$facetsOf K$ and \C$ridgesOf K$. If a facet $\sigma$ is represented by the variable \C$F$, the cone $a_\sigma$ is given by the term \C$coneOf F$. \Cref{cond:root} on $z^\star$ corresponds to the property \C$isKGeneric (d.+1) z$. (The use of the term \C$d.+1$ instead of \C$d$ is due to the shift by $1$ of the dimension of polyhedra and cones in \CoqPolyhedra{}.)

As explained in \Cref{sec:introduction}, the statement \C$odd_covering_theorem$ is still admitted. From this theorem, we formalize the proofs of the correctness of the vertex-set completeness certification (\Cref{th:vertex-completeness}) and that of the vertex-set certification (\Cref{cor:vertex-set}). This is done in the file \C$src/CertificateCorrectness.v$. The final statement there is:
\begin{lstlisting}
Theorem |*certificate_correctness*| cert :
    well_formedness cert
 -> full_dim_check cert
 -> feasibility_check cert    ((*(* T1 *)*))
 -> mapping_check cert        ((*(* T2 *)*))
 -> graph_check cert          ((*(* T3 *)*))
 -> inversibility_check cert  ((*(* T4 *)*))
 -> separability_check cert   ((*(* T5 *)*))
 -> flag_check cert           ((*(* T6 *)*))
 -> vertex_set '[polytope cert] = points cert.
\end{lstlisting}
where \C$'[polytope cert]$ and \C$points cert$ respectively extract the polytope $P$ and the set $V$ of points from the certificate. The correctness proof of the graph certification method (\Cref{th:graph-completeness,th:graph}) is still work-in-progress.

\subsection{Efficient Implementation and Refinement Proof}\label{subsec:refinement}\label{subsec:low-level}

The low-level checker (\C$src/LowLevelChecker.v$) implements the same
test functions as the high-level one, but using computation-oriented types: primitive $63$-bit integers and persistent arrays, and arbitrary-precision integers
and rationals of the \Rocq{} library \C$bignums$.\footnote{\url{https://github.com/rocq-community/bignums}} Arrays are used almost exclusively in a read-only way (the only exception is the BFS algorithm to compute graph diameters); therefore, we avoid the overhead associated with persistence.

This checker relies on a low-level
certificate that encodes the components of \Cref{sec:certificate}. An
interpretation function assigns to such a certificate the high-level
certificate it denotes, each rational denoting an element of the field, each array a finite set, each adjacency table a graph. The
certificate is stored as a compact binary serialization and is loaded inside
\Rocq{} by \C$binreader$, a plugin that allows to build low-level
records directly in memory from this format.\footnote{\url{https://github.com/Coq-Polyhedra/coq-binreader}}

The proof relating the low-lever and high-lever versions of the checker (\C$src/Refinement.v$) consists of two layers of a different nature. The lower layer is
relational: to each core functionality of the checker (bounded scans, folds,
binary search) corresponds a reflection lemma, stating that the boolean
computation holds \emph{if and only if} its mathematical counterpart holds.
The arithmetic tests on arbitrary-precision integers and rationals are
similarly related, in both directions, to the operations of the ordered field.
On the other hand, the upper layer is
refinement based: for each boolean check, we prove a one-directional
statement, namely that if the low-level check is validated, then the interpreted high-level
certificate satisfies the corresponding conditions of
\C$src/HighLevelCertificate.v$. This direction is the only
one that certification requires. Combined with the high-level statements proved in \C$src/CertificateCorrectness.v$., we obtain end-to-end statements such as
\begin{lstlisting}
Corollary |*vertex_equality*| cert :
     VtxContainment.check_certificate cert
  -> VtxEquality.check_certificate cert
  -> vertex_set '[hpoly_of cert] = points_of cert.
\end{lstlisting}
where \C$hpoly_of cert$ and \C$points_of cert$ are the polytope and the
point set described by the low-level certificate.

Finally, the extracted OCaml version of the checker (provided in the directory \C$ocaml$ of \C$homology-checker$) is obtained from the same
low-level functions, with primitive integers and arrays realized by native
OCaml ones, and arbitrary precision numbers handled by the \C$zarith$ library.\footnote{\url{https://ocaml.org/p/zarith/latest}}

\subsection{Experimental Results}\label{subsec:experiments}

We give a brief description of the instances used for the experiments, and refer to \Cref{app:experiments} for a more detailed presentation. The $n$-Birkhoff polytope (\C$birkhoff_$$n$) is the polytope of $n\times n$-doubly stochastic matrices. It has dimension $(n-1)^2$ and $n!$ vertices. The $d$-dimensional cross-polytope (\C$cross_$$d$) was already defined in \Cref{sec:odd-covering}, and its vertices are the $2d$ points $\pm e_i$ ($ i \in [d]$), where $(e_i)_{i\in[d]}$ denotes the canonical basis of $\R^d$. The $d$-dimensional cube (\C$cube_$$d$) is the polytope $[0,1]^d$. The $(d,n)$-dual cyclic polytope (\C$cyclic_$$d$\C$_$$n$) is known to maximize the number of vertices among $d$-dimensional polytopes with $n$ facets. The $(k,n)$-hypersimplex is defined as the set of points $x \in [0,1]^n$ satisfying $\sum_{i=1}^n x_i = k$. The instance \C$hypersimplex_$$d$ stands for the $(k,n)$-hypersimplex with $k = \big\lfloor\frac{d+1}{2}\big\rfloor$ and $n = d+1$. The permutahedron of order $n$ (instance \C$permutahedron_$$n$) is the convex hull of all permutations of the vector $(1,\dots,n)$. All these instances have (sub)exponentially many vertices as the associated parameter grows, except the cross-polytopes. The instances \C$poly20dim21$ and \C$poly23dim24$ are the two counterexamples to the Hirsch conjecture from~\cite{Matschke2015}. They are respectively $20$- and $23$-dimensional polytopes defined $40$ and $46$ inequalities with large rationals coefficients, and have $36\,425$ and $73\,224$ vertices. 

Benchmarks can be run from the \C$benchmarks$ directory of the project \C$homology-checker$ (we refer to the \C$README.md$ file for detailed instructions). Instances are provided in the subdirectory \C$benchmarks/data$ under the form of \lrs{} input files. The pipeline first invokes \lrs{} in order to generate the list of lex-feasible bases (which is part of the usual vertex enumeration procedure). The list is then post-processed by the certificate generator implemented in Rust to produce the certificate under the form of a binary file (project \C$lrs-postprocess$; see \Cref{sec:introduction}). The certificate generation follows the construction described in \Cref{subsec:completeness}. The time taken by the generation step is less than $20\%$ of the \lrs{} time (except for cubes where it is around $40\%$). The resulting certificate is finally loaded into \Rocq{} through the \C$binreader$ plugin. 

We consider two execution modes of the checker: one running directly inside \Rocq{} via its bytecode virtual machine (the \C$vm_compute$ tactic), and one obtained by extracting the low-level checker to OCaml and compiling it to native code. Measurements are carried out on a 12th~Gen Intel Core i7-12700H with 32\,GB of memory, with every benchmark pinned to a performance core, turbo mode and C-states disabled. For both sets of experiments, we report the \lrs{} vertex-enumeration time and the checker running time in seconds. We also report the checker time as a factor of the \lrs{} time (columns $\times$\lrs{}). This relative measure makes the comparison less dependent on the experimental setup and provides a direct measure of the cost of formal certification relatively to informal vertex enumeration. The \lrs{} reference times are single-threaded runs of \C$lrsgmp$, the GMP-based build of \lrs{}.

\Cref{tab:rocq-bench} reports the performance of the certification method inside \Rocq{} and compares it with that of~\cite{AllamigeonCanuStrub2023}. Since the latter work deals with certification of the vertex-edge graph only, the time reported for our checker corresponds to a single \C$vm_compute$ evaluation verifying \Cref{cond:feasibility,cond:facet_subset,cond:ridge_count,cond:rank,cond:sep,cond:flag,cond:incomparable-edges}. The time reported includes the loading time of the certificate, like in~\cite{AllamigeonCanuStrub2023}. For every instance, we also report the running time of the approach of~\cite{AllamigeonCanuStrub2023}, normalized by the \lrs{} time, as given in~\cite[Table~5.3]{Canu2025}. On the two counterexamples to the Hirsch conjecture, the evaluation additionnally includes the computation of the eccentricity of a distinguished vertex in the certified graph (\ie, the maximum graph distance from this vertex to any other vertex), which allows us to verify that the Hirsch bound is violated. The last column of \Cref{tab:rocq-bench} shows speedups ranging from $15\times$ to more than $1000\times$ over~\cite{AllamigeonCanuStrub2023}. Our checker also scales to larger cubes, cross-polytopes, dual cyclic polytopes, and permutahedra for which no timings were reported in their experiments. As extra measurements, we have run the certification method for the two counterexamples to the Hirsch conjecture (\C$poly20dim21$ and \C$poly23dim24$) on a standard laptop not configured for stable benchmarking (an Apple M4~Pro with $24$\,GB of memory). It respectively takes $31$ and $72$~seconds. This is to be compared with more than one and two hours reported in~\cite{AllamigeonCanuStrub2023}. The size of the certificates also benefit from a major improvement. For instance, the certificate for \C$poly20dim21$ with our method only takes $67\,\text{MB}$ vs more than $1.2\,\text{GB}$ in~\cite{AllamigeonCanuStrub2023}.

\begin{table}
\caption{%
  Performance of the \Rocq{} certification, compared to the approach
  of~\citet{AllamigeonCanuStrub2023}.
}
\label{tab:rocq-bench}
\centering\scriptsize\setlength{\tabcolsep}{3pt}
\begin{tabular}{lrrrrr}
\toprule
Instance & \lrs{} & \multicolumn{2}{c}{checker} & \cite{Canu2025} & \multicolumn{1}{c}{Improvement} \\
\cmidrule(lr){3-4}
 & (s) & (s) & $\times$\lrs{} & $\times$\lrs{} & \multicolumn{1}{c}{\raisebox{2pt}{factor}} \\
\midrule
poly20dim21              &    10.90 &    69.7 &   6.4 &   467 & 73$\times$ \\
poly23dim24              &    15.14 &   159.7 &  10.6 &   656 & 62$\times$ \\
\midrule
cube\_15                 &     1.12 &    34.6 &  31.1 &   471 & 15$\times$ \\
cube\_16                 &     2.41 &    78.0 &  32.3 &   515 & 16$\times$ \\
cube\_18                 &    11.75 &   368.9 &  31.4 &    -- & -- \\
cube\_19                 &    25.40 &   809.5 &  31.9 &    -- & -- \\
cube\_20                 &    55.33 &  1770.4 &  32.0 &    -- & -- \\
\midrule
cross\_6                 &     0.05 &     0.1 &   2.5 &   353 & 141$\times$ \\
cross\_7                 &     0.71 &     1.2 &   1.7 &   370 & 221$\times$ \\
cross\_8                 &    11.27 &    12.2 &   1.1 &  1185 & 1091$\times$ \\
cross\_9                 &   207.63 &   144.6 &   0.7 & $\dagger$ & $\dagger$ \\
cross\_10                &  4287.48 &  2123.8 &   0.5 & $\dagger$ & $\dagger$ \\
\midrule
cyclic\_9\_18            &     0.05 &     0.8 &  14.9 &   407 & 27$\times$ \\
cyclic\_10\_20           &     0.20 &     2.5 &  12.5 &   336 & 27$\times$ \\
cyclic\_11\_22           &     0.59 &     6.2 &  10.4 &   433 & 42$\times$ \\
cyclic\_12\_24           &     2.06 &    19.6 &   9.5 &   472 & 49$\times$ \\
cyclic\_13\_26           &     6.10 &    46.7 &   7.7 &   550 & 72$\times$ \\
cyclic\_14\_28           &    23.00 &   148.4 &   6.5 & $\dagger$ & $\dagger$ \\
cyclic\_15\_30           &    66.30 &   368.6 &   5.6 & $\dagger$ & $\dagger$ \\
cyclic\_16\_32           &   225.23 &  1249.5 &   5.5 & $\dagger$ & $\dagger$ \\
\midrule
permutahedron\_7         &     0.36 &     4.1 &  11.2 &   263 & 23$\times$ \\
permutahedron\_8         &     6.79 &    59.9 &   8.8 &   371 & 42$\times$ \\
permutahedron\_9         &   126.67 &  1086.2 &   8.6 & $\dagger$ & $\dagger$ \\
\midrule
hypersimplex\_17         &    19.05 &   316.8 &  16.6 &    -- & -- \\
hypersimplex\_18         &    44.30 &   721.2 &  16.3 &    -- & -- \\
hypersimplex\_19         &    95.54 &  1646.3 &  17.2 &    -- & -- \\
\midrule
birkhoff\_5              &     0.61 &     5.4 &   8.8 &    -- & -- \\
birkhoff\_6              &    69.72 &   634.0 &   9.1 &    -- & -- \\
\bottomrule
\end{tabular}

\par\smallskip
{\raggedright `--': instance not part of the benchmarks
of~\cite{Canu2025}. $\dagger$: instance beyond the reach at that
time.\par}
\end{table}

\Cref{tab:ocaml-bench} reports the performance of the checker extracted to OCaml,
with the same \lrs{} reference times. The loading time (parsing the
binary certificate into memory) is reported separately. The last two
columns give the time of the vertex-set certification alone (\Cref{cond:feasibility,cond:facet_subset,cond:ridge_count,cond:rank,cond:sep,cond:flag}), and of
the vertex-edge graph certification (\Cref{cond:feasibility,cond:facet_subset,cond:ridge_count,cond:rank,cond:sep,cond:flag,cond:incomparable-edges}), loading included in
both. The vertex-set certification is faster than the enumeration
on every instance, by a factor ranging from $1.7$ on the cubes to
$50$ on the cross-polytopes. The latter instances constitute extreme test cases, since, for the $d$-dimensional cross-polytope, the number of facets in the abstract simplicial complex (or equivalently, the number of lex-feasible bases) equals $2 d!$, which is exponentially larger than the number $2d$ of vertices. The speedup over \lrs{} illustrates the computational savings allowed by the combinatorial criterion of \Cref{th:odd-covering} over a purely numerical computation of lex-feasible bases in such a setting.\footnote{More realistically, cross-polytopes should be handled via the polar-based approach described in \Cref{remark:complexity}, which would allow to scale to larger dimensions like for cubes (the polar of the $d$-dimensional cross-polytope is the $d$-cube $[-1,1]^d$).}

Finally, regarding the performance of the graph certification, the full check remains faster than the enumeration (by a factor of~$1.25$ in the worst case).

\begin{table}
\caption{%
  Performance of the extracted checker, relative to the \lrs{} vertex
  enumeration time.
}
\label{tab:ocaml-bench}
\centering\footnotesize\setlength{\tabcolsep}{3pt}
\begin{tabular}{lrrrr}
\toprule
Instance & \lrs{} & \multicolumn{3}{c}{checker} \\
\cmidrule(lr){3-5}
 & (s) & load (s) & vertices & graph \\
 & & & $\times$\lrs{} & $\times$\lrs{} \\
\midrule
poly20dim21              &    10.90 &    0.26 &  0.16 &  0.19 \\
poly23dim24              &    15.14 &    0.51 &  0.18 &  0.24 \\
\midrule
cube\_20                 &    55.33 &    6.61 &  0.58 &  0.76 \\
cube\_21                 &   118.86 &   13.69 &  0.58 &  0.76 \\
cube\_22                 &   254.51 &   28.20 &  0.58 &  0.76 \\
\midrule
cross\_8                 &    11.27 &    0.08 &  0.03 &  0.03 \\
cross\_9                 &   207.63 &    0.62 &  0.02 &  0.02 \\
cross\_10                &  4287.48 &    7.64 &  0.02 &  0.02 \\
\midrule
cyclic\_16\_32           &   225.23 &    5.43 &  0.17 &  0.20 \\
cyclic\_17\_34           &   708.99 &   12.69 &  0.15 &  0.17 \\
cyclic\_18\_36           &  2535.86 &   40.09 &  0.13 &  0.16 \\
\midrule
permutahedron\_8         &     6.79 &    0.14 &  0.23 &  0.25 \\
permutahedron\_9         &   126.67 &    1.33 &  0.22 &  0.23 \\
permutahedron\_10        &  2958.84 &   15.69 &  0.19 &  0.19 \\
\midrule
hypersimplex\_19         &    95.54 &    4.37 &  0.23 &  0.46 \\
hypersimplex\_20         &   217.99 &    9.33 &  0.22 &  0.45 \\
hypersimplex\_21         &   466.56 &   19.63 &  0.23 &  0.47 \\
\midrule
birkhoff\_5              &     0.61 &    0.02 &  0.16 &  0.21 \\
birkhoff\_6              &    69.72 &    1.51 &  0.15 &  0.20 \\
\bottomrule
\end{tabular}

\end{table}

\section{Conclusion and Future Work}\label{sec:conclusion}

We have presented a new method for certifying vertex sets and vertex-edge graphs of polyhedra. It is based on an odd-covering criterion which relaxes the conditions required to certify a triangulation of the normal fan. We have implemented the resulting checker in the proof assistant \Rocq{} and formally proved its correctness up to two ingredients: the odd-covering criterion and the correctness argument underlying to vertex-edge graph certification. For both, we provide self-contained pen-and-paper proofs in this paper, leaving their formalization for future work. We have evaluated the method on several representative families of polyhedra, including ones involved in the disproof of the Hirsch conjecture. We have shown that our method outperforms the previous work of \cite{AllamigeonCanuStrub2023} by several order of magnitudes, and that the OCaml version of the checker obtained by extraction from the implementation in \Rocq{} is faster than the state-of-the-art library \lrs{} for vertex enumeration. These results show that formal guarantees for one of the most fundamental operations on polyhedra can now be obtained at only a modest additional cost compared to informal computation. The mathematical principles underlying our approach are not specific to \Rocq{}, and could be used in other proof assistants.

This work opens several perspectives. First, we plan to extend the certification method to arbitrary polyhedra (taking into account vertices and extreme rays), and to address the converse V-to-H conversion problem. This does not require any fundamentally new ideas. The treatment of general polyhedra can be done by that of pointed polyhedral cones and their extreme rays, which is equivalent to the vertex enumeration of polytopes (by compactification). Similarly, the V-to-H conversion is simply the dual counterpart of the H-to-V problem considered here, as outlined in \Cref{remark:complexity}.

On the formalization side, we would like to substitute the current refinement proof relating the high-level and low-level checkers (generated with the assistance of AI) by a proof obtained using dedicated tools for refinements, such as the one described in~\cite{CohenCranceMahboubi2024}. The implementation could also benefit from a tighter integration between \lrs{} and the certificate checker. The computations done within \lrs{} could be indeed directly used to produce a substantial part of the certificate. Besides, an improvement of the pipeline between \lrs{} and the checker could remove a part of the time currently taken by the generation and loading of binary certificate files. 

From a theoretical perspective, we would also like to replace \Cref{cond:incomparable-edges} on edges by a criterion based on flags. Given the possibly large number of edges, this requires identifying a compact representation in which flags are shared.

\begin{acks} 
Part of this work was done while the second author was at Inria and CMAP, Ecole polytechnique. The first author thanks Frédéric Meunier for many helpful discussions and valuable pointers on simplicial homology.  
\end{acks}

\bibliographystyle{ACM-Reference-Format}

\appendix

\section{Properties of the Normal Fan}\label{app:normal-fan}

We review basic properties on normal cones used thereafter. To begin with, we recall that the relative interior $\relint Q$ of a polyhedron $Q$ is defined as the interior relatively to its affine hull, and that
\begin{equation}\label{eq:relint}
\relint Q = Q \setminus \bigg(\bigcup_{\substack{F \; \text{face of} \; Q\\ \dim F < \dim Q}} F \bigg) = Q \setminus \bigg(\bigcup_{F \; \text{facets of} \; Q} F \bigg) \, .
\end{equation}
Moreover, we recall that a convex set $S \subset \R^d$ has dimension $d$ if and only if it has nonempty interior~\cite{Rockafellar1970}.

\begin{proposition}\label{prop:normal-cone-def}
Suppose that $P$ is defined by the inequalities $\scalar{a_i}{x} \leq b_i$ where $a_i \in \R^d$ and $b_i \in \R$ for all $i \in [n]$. Then, for all $x \in P$,
\[
N(x) = \cone{\{a_i\}_{i \in I(x)} \quad \text{where} \enspace I(x) \coloneqq \{ i \in [n] \colon \scalar{a_i}{x} = b_i \}} \, .
\]
\end{proposition}

\begin{proof}
The inclusion $\supset$ is a direct consequence of the fact that, for all $i \in I(x)$, $\scalar{a_i}{x} = b_i \geq \scalar{a_i}{y}$. For the symmetric inclusion, take $z \in N(x)$. By Strong Duality Theorem (that we can apply $P$ is nonempty), there exists $(\lambda_i)_i \in \R_{\geq 0}^n$ such that $z = \sum_{i = 1}^n \lambda_i a_i$, and, by the complementary slackness conditions, $\lambda_i = 0$ or $\scalar{a_i}{x} = b_i$ for all $i \in [n]$. We deduce that $\lambda_i = 0$ for all $i \in [n] \setminus I(x)$, which completes the proof.
\end{proof}

Normal cones apply more generally to faces of $P$. Given a nonempty face $Q$ of $P$, we set $N(Q) \coloneqq N(x)$ where $x$ is an arbitrary point of $\relint Q$. This definition is valid since the set $I(x)$ is constant over $x \in \relint Q$. We analogously set $I(Q) \coloneqq I(x)$ for $x \in \relint Q$. 

\begin{proposition}\label{prop:normal-cone-basic-prop}
\begin{enumerate}[label=(\roman*), ref=(\roman*)]
\item\label{item:i} For each face $Q$ of $P$, $\dim N(Q) = d - \dim Q$.
\item\label{item:ii} If $Q, Q'$ are two nonempty faces of $P$ such that $Q' \subset Q$, then $N(Q)$ is a face $N(Q')$.
\item\label{item:iii} If $Q$ be a nonempty face of $P$, and $z \in \relint N(Q)$, the set of maximizers of the linear function $x \mapsto \scalar{z}{x}$ over $P$ is precisely the face $Q$.
\end{enumerate}
\end{proposition}

\begin{proof}
\begin{enumerate}[label=(\roman*), ref=(\roman*)]
\item This is a consequence of the fact that $\dim N(Q) = \rank{\{a_i\}_{i \in I(Q)}} = d - \dim Q$ (see~\cite[Chapter~8]{Schrijver86})

\item We have $N(Q) \subset N(Q')$. As the normal cones form a fan, $N(Q) \cap N(Q') = N(Q)$ must be a face of $N(Q')$.

\item Let $Q'$ be the face formed by the maximizers of the linear function $x \mapsto \scalar{z}{x}$ over $P$. As $z \in N(Q)$, we have $Q \subset Q'$. Suppose that $Q \subsetneq Q'$, so that $\dim Q < \dim Q'$~\cite[Theorem~2.7]{ziegler}. Then $Q$ is a proper face of $Q'$, and thus $N(Q')$ is a proper face of $N(Q)$ by~\ref{item:i} and~\ref{item:ii}. Thus, $N(Q') \cap \relint N(Q) = \varnothing$, which contradicts the fact that $z \in N(Q') \cap \relint N(Q)$.
\qedhere
\end{enumerate}
\end{proof}

\begin{corollary}\label{cor:normal-cone-basic-prop}
For every nonempty face $Q$ of $P$ and $w \in P$, $(\relint N(Q)) \cap N(w) = \varnothing$ unless $w \in Q$. 
\end{corollary}

\begin{proof}
Let $z \in (\relint N(Q)) \cap N(w)$. Then $w$ belongs to the set of maximizers of $x \mapsto \scalar{z}{x}$ over $P$, which is exactly the face $Q$ by \Cref{prop:normal-cone-basic-prop}\ref{item:iii}.
\end{proof}

In the next statement, we recall that the recession cone of $P$ is the cone defined as
\[
\rec P \coloneqq \{ y \in \R^d \colon x^\circ + \R_{\geq 0} y \subset P \}
\]
where $x^\circ$ is an arbitrary element of $P$. The polar cone of $\rec P$ is defined as
\[
(\rec P)^\circ \coloneqq \{ z \in \R^d \colon \forall y \in \rec P \, , \; \scalar{y}{z} \leq 0 \} \, .
\]
Equivalently, $(\rec P)^\circ$ consists of the vectors $z \in \R^d$ such that the supremum of $x \mapsto \scalar{z}{x}$ over $P$ is finite.

\begin{proposition}\label{prop:normal-cone-support}
The support of the normal fan of $P$, \ie, the union of all normal cones, is the polar cone of the recession cone of $P$. In particular, it is equal to $\R^d$ if and only if $P$ is a polytope.   
\end{proposition}

\begin{proof}[\proofname{} of \Cref{prop:vertex-completeness}]
Suppose that $P$ is a polytope and $V \supset \vertices(P)$. By~\Cref{prop:normal-cone-basic-prop}\ref{item:ii}, the normal cones $N(v)$ for $v \in \vertices P$ are maximal inclusionwise (among all normal cones), so that the support of the normal fan is the union of these cones. We deduce from~\Cref{prop:normal-cone-support} and $V \supset \vertices(P)$ that $\bigcup_{v \in V} N(v) = \R^d$.

Conversely, suppose that $\bigcup_{v \in V} N(v) = \R^d$. By \Cref{prop:normal-cone-support}, $P$ is a polytope. Let $v \in \vertices(P)$, and take $z \in \interior N(v)$. Let $w \in V$ such that $z \in N(w)$. Then $z \in (\interior N(v)) \cap N(w)$, which implies $v = w$ by~\Cref{cor:normal-cone-basic-prop}. We conclude that $v \in V$.
\end{proof}

\section{Odd-Covering Criterion: Extra Statements and Proofs}\label{app:odd-covering}

Given an abstract simplicial complex $\Ksf$, we denote by $\Ksf^{(p)}$ the simplices of $\Ksf$ of dimension $p$. In the case of a pure $(d-1)$-dimensional abstract simplicial complex, $\facets$ represents the set of faces, and $\ridges$ the set of ridges. 

\begin{lemma}\label{lemma:rel-boundary}
Let $\sigma \in \facets$ such that $\dim a_\sigma = d$. Then, the $k$-dimensional faces of $a_\sigma$ are precisely the cones $a_\eta$ such that $\eta \subset \sigma$ and $\# \eta = k$. Moreover, $\partial a_\sigma = \bigcup_{\tau \in \partial \sigma} a_\tau$.
\end{lemma}

\begin{proof}
The first part is a consequence of the fact that $a_\sigma$ is a simplicial cone. By~\eqref{eq:relint}, the boundary of $a_\sigma$ is the union of its facets, which are precisely the cones $a_\tau$ for $\tau \in \partial \sigma$.
\end{proof}

\begin{lemma}\label{lemma:open-dense}
The set of $d$-generic vectors is open dense.
\end{lemma}

\begin{proof}
For every $\sigma \in \Ksf$ satisfying $\dim a_\sigma < d$, the complement set $a_\sigma$ is an open dense set. The set of $d$-generic vectors is the intersection of these complement sets.
\end{proof}

\begin{lemma}\label{lemma:finite-breakpoints}
Let $\gamma$ be a generic path. Then $\gamma(t)$ belongs to a $(d-1)$-dimensional cone $a_\sigma$ ($\sigma \in \Ksf$) only at finitely many values of $t \in [0,1]$.
\end{lemma}

\begin{proof}
The set $\{t \in [0,1] \colon \gamma(t) \in a_\tau \, , \; \dim a_\tau = d-1 \, , \;  \tau \in \Ksf \}$ is discrete, as each of its points is isolated by the transversality property. Moreover, it is closed (as the $a_\tau$ are closed and $\gamma$ is continuous). Thus, as a discrete compact set, it is finite.
\end{proof}

\begin{lemma}\label{lemma:linear-path}
Given two $d$-generic vectors $z$ and $u$, a linear path $\gamma \colon t \in [0,1] \mapsto z + t (u - z)$ is generic as soon as:
\begin{conditions}
\item\label{cond:small-dim} $z \notin \linspan\big(\{u\} \cup \{a_i\}_{i \in \eta}\big)$ for all $\eta \in \Ksf$ such that $\dim a_\eta < d-1$;
\item\label{cond:codim1} $z \notin u + \linspan{\{a_i\}_{i \in \tau}}$ for all $\tau \in \Ksf$ such that $\dim a_\tau = d-1$. 
\end{conditions}
\end{lemma}

\begin{proof}
First, let $\eta \in \Ksf$ such that $\dim a_\eta < d-1$. By \Cref{cond:small-dim}, we have $z + \alpha u \notin a_\eta$ for all $\alpha \in \R$, since $a_\eta \subset \linspan \{a_i\}_{i \in \eta}$. Thus, for all $t \in [0,1)$, $z + \frac{t}{1-t} u \notin a_\eta$. As $a_\eta$ is a cone and $u \notin a_\eta$ (recall that $u$ is $d$-generic), we deduce that $z + t(u-z) \notin a_\eta$ for all $t \in [0,1]$.

Second, let $\tau \in \Ksf$ such that $\dim a_\tau = d-1$. Suppose that the path $\gamma$ intersects $a_\eta$. As a linear path, the tangential direction is the same at any point of the path and equal to $u - z$, and the intersection is transversal if and only if $u - z \notin \linspan {\{a_i\}_{i \in \tau}}$, which is exactly \Cref{cond:codim1}.
\end{proof}

\begin{lemma}\label{lemma:non-full-dim-cones}
Let $\sigma \in \facets$ such that $a_\sigma$ is a pointed cone of dimension less than $d$, and $x$ a $(d-1)$-generic point in $a_\sigma$. Then, there are precisely two cones $a_\tau$ for $\tau \in \partial \sigma$ containing $x$. 
\end{lemma}

\begin{proof}
We must have $\dim a_\sigma = d-1$ as $x$ is $(d-1)$-generic and $x \in a_\sigma$ by assumption. By Carathéodory Theorem, there exists $\tau \in \ridges$ such that $\tau \subset \sigma$ and $x \in a_\tau$. Thus, $\tau \in \partial \sigma$, and $\dim a_\tau = d-1$ (as $x$ is $(d-1)$-generic). Let $j$ be unique element of $[n]$ such that $\sigma = \tau \cup \{j\}$. As $\dim a_\sigma = \dim a_\tau = d-1$, we have $a_j \in \linspan{\{a_i\}_{i \in \tau}}$, thus there exist scalars $\lambda_i \in \R$ ($i \in \tau$) such that $a_j = \sum_{i \in \tau} \lambda_i a_i$. As $a_j \neq 0$ and $a_\sigma$ is pointed, there exists at least one $i \in \tau$ such that $\lambda_i > 0$. Let $\mu_i \in \R_{\geq 0}$ ($i \in \tau$) such that $x = \sum_{i \in \tau} \mu_i a_i$. Then, for all $\alpha \in \R$, we have
\begin{equation}\label{eq:simplex-pivot}
x = \sum_{i \in \tau} (\mu_i - \alpha \lambda_i) a_i + \alpha a_j \, .
\end{equation}
Taking $k \in \arg\min_{i \in \tau \colon \lambda_i > 0} \frac{\mu_i}{\lambda_i}$ and $\alpha = \frac{\mu_k}{\lambda_k}$, we deduce that $x \in a_{\tau'}$ for $\tau' = \sigma \setminus \{k\}$. Hence, $\tau' \in \partial \Ksf$ and $\tau' \neq \tau$. 

Conversely, let $\tau'' \in \partial \sigma$ such that $\tau \neq \tau''$ and $x \in a_{\tau''}$. We must have $j \in \tau''$ since otherwise, $\tau'' \subset \tau$ while $\# \tau'' = \# \tau = d-1$. Let $\nu_i \in \R_{\geq 0}$ ($i \in \tau''$) such that $x = \sum_{i \in \tau''} \nu_i a_i$. We have $\nu_j > 0$ (otherwise, $x \in a_{\tau'' \setminus \{j\}}$ which is impossible as $x$ is $(d-1)$-generic). As the scalars $(\lambda_i)_{i \in \tau}$ are uniquely defined (since $\dim a_\tau = \rank \{a_i\}_{i \in \tau} = d-1$), the relation~\eqref{eq:simplex-pivot} holds for $\alpha = \nu_j$. We deduce that there exists $l \in \tau$ such that $\mu_l = \nu_j \lambda_l$. Hence, $l \in  \arg\min_{i \in \tau \colon \lambda_i > 0} \frac{\mu_i}{\lambda_i}$. In the case where $k \neq l$, we would obtain that $x \in a_{\tau \setminus \{k,l\} \cup \{j\}}$, which is impossible by the $(d-1)$-genericity of $x$. We conclude that $k = l$ and $\tau' = \tau''$.
\end{proof}

\begin{proof}[\proofname{} of \Cref{th:odd-covering}]
Let $u \in \R^d$ be a $d$-generic vector. Let $S^0$ (resp.~$S^1$) be the set of $\sigma \in \facets$ such that $z^*$ (resp.~$u$) belongs to $a_\sigma$ (or, equivalently, $\interior(a_\sigma)$, as both vectors are generic). 

We note that the conditions of \Cref{lemma:linear-path} can be fulfilled for $z = z^\star$ up to perturbing the vector $z^\star$ (since the set of vectors $z$ satisfying the conditions is dense). On top of making the linear path $\gamma \colon t \in [0,1] \mapsto z^\star + t (u - z^\star)$ generic, such a perturbation preserves the $d$-genericity of $z^\star$ (\Cref{lemma:open-dense}), leaves the set $S^0$ unchanged (since $\bigcap_{\sigma \in S^0} \interior(a_\sigma) \cap \bigcap_{\sigma \notin S^0} (\R^d \setminus a_\sigma)$ is an open set containing $z^\star$, and any point in it belongs to precisely the cones $a_\sigma$ for $\sigma \in S^0$).

We set $S(t) \coloneqq \{ \sigma \in \facets \colon \; \gamma(t) \in a_\sigma \}$. Let $t^\circ \in [0,1]$ such that $\gamma(t^\circ)$ is $(d-1)$-generic, and $t^-, t^+ \in [0,1]$ such that $t^- < t^\circ < t^+$ and $\gamma(t)$ is $d$-generic over $[t^-, t^\circ)$ and $(t^\circ, t^+]$. Then $S(t)$ is constant over $(t^-, t^\circ)$ and $(t^\circ, t^+)$. Indeed, over these intervals, $\gamma(t)$ only meet $d$-dimensional cones (as $\gamma(t)$ is $d$-generic), and it cannot meet their boundaries, since the latter are included in the set $\bigcup_{\tau \in \ridges} a_\tau$ (\Cref{lemma:rel-boundary}) and $\dim a_\tau < d$ for all $\tau \in \ridges$. 

Let us show that $\# S(t^-)$ and $\# S(t^+)$ have the same parity. To this purpose, we set $T \coloneqq \{ \tau \in \ridges \colon \gamma(t^\circ) \in a_\tau\}$. We claim that 
\begin{equation}\label{eq:symdiff}
S(t^-) \symdiff S(t^+) = 
\bigcup_{\tau \in T} \{ \sigma \in \facets \colon \tau \in \partial \sigma \, , \; \dim a_\sigma = d \big\} \, . 
\end{equation}
Indeed, $\sigma \in S(t^-) \symdiff S(t^+)$ implies that $\dim a_\sigma = d$ (as $\gamma(t^\pm)$ is $d$-generic) and $\gamma(t^\circ) \in \partial a_\sigma$, so that, by~\Cref{lemma:rel-boundary}, $\gamma(t^\circ) \in a_\tau$ for some $\tau \in \partial \sigma$. Conversely, let $\tau \in T$ and $\sigma \in \facets$ such that $\tau \in \partial \sigma$ and $\dim a_\sigma = d$. Firstly, as $a_\sigma$ is a simplicial cone, $a_\tau$ is a facet of $a_\sigma$. Secondly, $\gamma(t^\circ) \in \relint a_\tau$ as $\gamma(t^\circ)$ is $(d-1)$-generic (owing to~\eqref{eq:relint}). Thirdly, the path $\gamma$ transversally crosses $a_\tau$ at $\gamma(t^\circ)$. As a consequence, when $t$ grows from $t^-$ to $t^+$, $\gamma(t)$ either leaves or enters the cone $a_\sigma$ at $t = t^\circ$. We deduce that $\sigma \in S(t^-) \symdiff S(t^+)$.

The union in~\eqref{eq:symdiff} is disjoint. Indeed, for all $\sigma \in \facets$ such that $\dim a_\sigma = d$, we cannot have $\gamma(t^\circ) \in a_\tau \cap a_{\tau'}$ for two distinct $\tau, \tau' \in \partial \sigma$, since $a_\tau \cap a_{\tau'} = a_{\tau \cap \tau'}$ by the simpliciality of $a_\sigma$, and $a_{\tau \cap \tau'}$ is a face of dimension less than $d-1$ while $\gamma(t^\circ)$ is $(d-1)$-generic. Thus, we obtain that
\begin{align*}
& \# (S(t^-) \symdiff S(t^+)) \\
& = \sum_{\tau \in T} \# \{ \sigma \in \facets \colon \tau \in \partial \sigma \, , \; \dim a_\sigma = d \big\} \\
& = \sum_{\tau \in T} \# \{ \sigma \in \facets \colon \tau \in \partial \sigma \, , \; \dim a_\sigma < d \big\} \mod 2 \reasonnext{by \Cref{cond:ridge}} \\
& = \sum_{\substack{\sigma \in \facets \\ \dim a_\sigma < d}} \# (\partial \sigma \cap T) \mod 2 \\
& = 0 \mod 2 \reason{by \Cref{lemma:non-full-dim-cones} and \Cref{cond:pointed}.}  
\end{align*}

The previous argument shows that $\# S(0)$ and $\# S(1)$ have the same parity. Moreover, $S^0 = S(0)$ and $S^1 = S(1)$. Consequently, as $\# S^0$ is odd (\Cref{cond:root}), so is $\# S^1$.
\end{proof}

\begin{lemma}\label{lemma:local-homology-isomorphic}
Let $n \geq 1$ and $y \in S^n$. The inclusion map from $H_n(S^n; \F_2)$ to $H_n(S^n, S^n \setminus \{y\}; \F_2)$ is an isomorphism.
\end{lemma}

\begin{proof}
We recall the following exact sequence 
\begin{multline*}
H_n(S^n \setminus \{y\}; \F_2) \to H_n(S^n; \F_2) \to H_n(S^n, S^n \setminus \{y\}; \F_2) \\
\to H_{n-1}(S^n \setminus \{y\}; \F_2) \to H_{n-1}(S^n; \F_2) \, .
\end{multline*}
See~\cite[p.~115]{Hatcher}. The homology group $H_n(S^n \setminus \{y\}; \F_2)$ is trivial as $S^n \setminus \{y\}$ is contractible. This shows that the second inclusion map is injective. The last inclusion is an isomorphism (if $n \geq 2$, the two groups are trivial, while if $n = 1$, both are isomorphic $\F_2$). We deduce that the third map is trivial, which ensures that the second map is surjective.
\end{proof}

\begin{lemma}\label{lemma:local-homology-simplex}
Suppose $d \geq 2$. Let $\sigma$ be a facet of $\Ksf$, and $y$ be a $d$-generic point of $S^{d-1}$. Then $f_\sharp(\sigma)$ is a local cycle (w.r.t.~$y$), and $[f_\sharp(\sigma)]_y \neq 0$ if and only if $y \in a_\sigma$. 
\end{lemma}

\begin{proof}
We recall that $(f \circ \sigma)(\triangle_{d-1})$ is the spherical polytope $D_\sigma \coloneqq a_\sigma \cap S^{d-1}$. Moreover, $\partial(f_\sharp(\sigma)) = f_\sharp(\partial \sigma) = \sum_{\tau \in \partial \sigma} f \circ \tau$, and $\bigcup_{\tau \in \partial \sigma} (f \circ \tau)(\triangle_{d-2}) = \partial a_\sigma \cap S^{d-1} = \partial D_\sigma$ by \Cref{lemma:rel-boundary}. As $y$ is $d$-generic, $y \notin \partial a_\sigma$, and so $y \notin \partial D_\sigma$. We deduce that $\partial(f_\sharp(\sigma)) \in C_{d-2}(S^{d-1} \setminus \{y\}; \F_2)$. Thus, $f_\sharp(\sigma)$ is a local cycle w.r.t.~$y$. 

By genericity of $y$, either $y \notin a_\sigma$ or $y \in \interior(a_\sigma)$. In the first case, $f_\sharp(\sigma)$ is a chain of $C_{d-1}(S^{d-1} \setminus \{y\}; \F_2)$ and so $[f_\sharp(\sigma)]_y = 0$. In the second case, $D_\sigma$ is homeomorphic to a $(d-1)$-simplex. The inclusion maps then induce the following isomorphisms:
\begin{multline*}
H_{d-1}(D_\sigma, \partial D_\sigma; \F_2) \approx H_{d-1}(D_\sigma, D_\sigma \setminus \{y\}; \F_2) \\
\approx H_{d-1}(S^{d-1}, S^{d-1} \setminus \{y\}; \F_2) 
\end{multline*}
where the first one follows from the fact that $\partial D_\sigma$ is a retract of $D_\sigma \setminus \{y\}$, and the second one from the Excision Theorem (see~\cite[Th.~2.20]{Hatcher} with $X = S^{d-1}$, $A = S^{d-1} \setminus \{y\}$ and $B = D_\sigma$). All groups are isomorphic to $\F_2$ (by \Cref{lemma:local-homology-isomorphic}), and the first one is generated by the class of $f_\sharp(\sigma)$. We deduce that $H_{d-1}(S^{d-1}, S^{d-1} \setminus \{y\}; \F_2)$ is generated by $[f_\sharp(\sigma)]_y \neq 0$.
\end{proof}

\section{Certification Method: Extra Statements and Proofs}\label{app:certificate}

\begin{proof}[\proofname{} of \Cref{lemma:feasibility}]
The first part of the statement is immediate from \Cref{cond:feasibility}. The second part follows from the characterization of $N(v)$ given in \eqref{eq:normal-cone-characterization} and the fact that $I_v = I(v)$.
\end{proof}

\begin{proof}[\proofname{} of \Cref{lemma:facet-subset-normal-cone}]
Let $\sigma \in F$. By \Cref{lemma:feasibility}, we have $\phi(v) \in P$. By \Cref{cond:facet_subset}, we have $\sigma \subset I_{\phi(\sigma)}$, hence 
\[
a_\sigma \subset \cone{\{a_i\}_{i \in I_{\phi(\sigma)}}} = N(\phi(\sigma))
\] 
by \Cref{lemma:feasibility} again and the characterization~\eqref{eq:normal-cone-characterization}.
\end{proof}

When \Cref{cond:ridge_count} holds, for every $\sigma \in F$ and $i \in \sigma$, we set $\sigma^{(i)}$ to an arbitrary neighbor $\sigma'$ of $\sigma$ satisfying $\sigma' \neq \sigma$ and $\sigma \setminus \{i\} \subset \sigma'$.

\begin{lemma}\label{lemma:neighbors}
Under \Cref{cond:ridge_count}, each $\sigma \in F$ has exactly $d$ neighbors in $G$, given by the $\sigma^{(i)}$ for $i \in \sigma$.
\end{lemma}

\begin{proof}
Let $\sigma \in F$. We claim that the $\sigma^{(i)}$ for $i \in \sigma$ are pairwise distinct. By contradiction, suppose that $\sigma \setminus \{j\} \subset \sigma^{(i)}$ for some $j \in \sigma$, $j \neq i$. As $i \in \sigma \setminus \{j\}$, we get that $i \in \sigma^{(i)}$. Thus, $\sigma \subset \sigma^{(i)}$, which leads to $\sigma = \sigma^{(i)}$ since both sets have cardinality $d$, a contradiction.

We deduce that $\sigma$ has at least $d$ pairwise distinct neighbors $\sigma^{(i)}$ ($i \in \sigma$). As the degree of $\sigma$ in $G$ is at most $d$, the $\sigma^{(i)}$ are precisely the $d$ neighbors of $\sigma$.
\end{proof}

\begin{lemma}\label{lemma:simple}
The graph $G$ is simple.
\end{lemma}

\begin{proof}
The graph $G$ has no loop as, for all $\sigma \in F$, its neighbors are the $\sigma^{(i)}$ for $i \in \sigma$ (\Cref{lemma:neighbors}), and $\sigma \neq \sigma^{(i)}$. Moreover, given $\sigma \in F$, there are no multiple edges incident to $\sigma$ as $\sigma$ has exactly $d$ neighbors (\Cref{lemma:neighbors}) and degree at most $d$.
\end{proof}

\begin{proof}[\proofname{} of \Cref{lemma:ridge-even}]
We label every edge of $G$ between $\sigma$ and $\sigma'$ by the set $\sigma \cap \sigma'$. 

Let $\tau \in \ridges$, and $S \coloneqq \{ \sigma \in \facets \colon \tau \subset \sigma \}$. We claim that $S$ is precisely the set of vertices of the edges of $G$ labeled by $\tau$. Indeed, let $\sigma \in S$ and $i \in \sigma$ such that $\tau = \sigma \setminus \{i\}$. We have $\tau = \sigma \setminus \{i\} \subset \sigma^{(i)}$. Thus, $\sigma \cap \sigma^{(i)} = \tau$ (as $\sigma \cap \sigma^{(i)} \supset \tau$, and $\sigma \cap \sigma^{(i)}$ has cardinality $d-1$ since $\sigma \neq \sigma^{(i)}$), so the edge between $\sigma$ and $\sigma'$ is labeled by $\tau$. Conversely, if an edge between $\sigma$ and $\sigma'$ of $G$ is labeled by $\tau$, then, by definition of the labels, $\tau = \sigma \cap \sigma'$, thus $\tau \subset \sigma, \sigma'$ and, subsequently, $\sigma, \sigma' \in S$. This proves the claim.

We now claim that the edges labeled by $\tau$ form a matching, \ie, a set of edges whose vertices are pairwise disjoint. Indeed, by \Cref{lemma:neighbors}, every vertex $\sigma$ in $G$ is incident to precisely $d$ edges labeled by $\sigma \cap \sigma^{(i)} = \sigma \setminus \{i\}$ for $i \in \sigma$. Thus, two edges labeled by a ridge $\tau$ cannot intersect (recall that there are no multiple edges by \Cref{lemma:simple}).

As a consequence, $S$ has even cardinality, as the set of the vertices of a matching. 
\end{proof}

\begin{proof}[\proofname{} of \Cref{prop:root}]
By \Cref{cond:rank}, setting $\sigma^\star \coloneqq (i_1, \dots, i_d)$, the matrix $\begin{pmatrix} a_{i_1} & \dots & a_{i_d} \end{pmatrix}^\top \begin{pmatrix} f_1 & \dots & f_d \end{pmatrix}$ is a diagonal matrix with positive diagonal entries. Thus, the family of vectors $a_i$ for $i \in \sigma^\star$ has rank $d$. Consequently, $a_{\sigma^\star}$ is a $d$-dimensional (simplicial) cone. This implies that $N(v^\star)$ has dimension $d$ (as it contains $a_{\sigma^\star}$ by \Cref{lemma:facet-subset-normal-cone}), and thus $v^\star$ is a vertex of $P$ (\Cref{prop:normal-cone-basic-prop}\ref{item:i}).

The vector $z^\star$ obviously belongs to $a_{\sigma^\star}$; in fact, it belongs to its interior, as a positive combination of all $a_i$ for $i \in \sigma^\star$, and thanks to the fact that these $a_i$ form a family of rank $d$. Second, let $\sigma \in F$, $\sigma \neq \sigma^\star$. We first deal with the case where $\phi(\sigma) = v^\star$. In this case, we have $\scalar[\big]{x}{\sum_j (s_\sigma)_j f_j} \leq 0$ for all $x \in a_\sigma$, thanks to \Cref{cond:sep}. In contrast, $\scalar{z^\star}{f_j} > 0$ for all $j \in [d]$ (as a consequence of \Cref{cond:rank}), and, since  $s_\sigma \in \R_{\geq 0}^d \setminus \{0\}$ (\Cref{cond:sep} again), we get that $\scalar[\big]{z^\star}{\sum_j (s_\sigma)_j f_j} > 0$. Thus, $z^\star \notin a_\sigma$. It remains to consider the case where $w \coloneqq \phi(\sigma) \neq v^\star$. Recall that $a_\sigma \subset N(w)$ (\Cref{lemma:facet-subset-normal-cone}), and $(\interior N(v^\star)) \cap N(w) = \varnothing$ as $w \neq v^\star$ (\Cref{cor:normal-cone-basic-prop}). We deduce that $z^\star \notin a_\sigma$.

Finally, we argue that $z^\star$ is a $d$-generic vector. To this purpose, it suffices to prove that $a_{\sigma^\star}$ is the only cone $a_\eta$ for $\eta \in \Ksf$ containing $z^\star$ (since $a_{\sigma^\star}$ is $d$-dimensional). By contradiction, let $\eta \in \Ksf$, $\eta \neq \sigma^\star$. If $\eta \subset \sigma$ for some $\sigma \in F \setminus \{\sigma^\star\}$, then $z^\star \notin a_{\eta}$ as $a_\eta \subset a_\sigma$ and $z^\star \notin a_\sigma$ by the previous argument. Suppose now that $\eta \subsetneq \sigma^\star$. As $a_{\sigma^\star}$ is a $d$-dimensional simplicial cone, the cone $a_\eta$ is a face of $a_{\sigma^\star}$ of dimension less than $d$, and thus it is included in the boundary of $a_{\sigma^\star}$. However, $z^\star \in \interior(a_{\sigma^\star})$, hence $z^\star \notin a_\eta$.
\end{proof}

\begin{lemma}\label{lemma:normal-cone-pointed}
Suppose that $P$ is full-dimensional polyhedron. For every point $x \in P$, the cone $N(x)$ is pointed.
\end{lemma}

\begin{proof}
Suppose that $N(x)$ is not pointed, and take $z \in N(x) \cap (-N(x))$, $z \neq 0$. Set $\alpha \coloneqq \scalar{z}{x}$. Then, for all $y \in P$, $\scalar{z}{y} \leq \alpha$ and $-\scalar{z}{y} \leq -\alpha$. Thus, $P$ is contained in the hyperplane $\{y \in \R^d \colon \scalar{z}{y} = \alpha \}$, and thus, its dimension is less than $d$. 
\end{proof}

\begin{proof}[\proofname{} of \Cref{th:vertex-completeness}]
\Cref{cond:ridge} of \Cref{th:odd-covering} holds thanks to \Cref{lemma:ridge-even}, and \Cref{cond:root} thanks to \Cref{prop:root}. As the polyhedron $P$ has dimension $d$, $N(v)$ is pointed for all $v \in V$ (\Cref{lemma:normal-cone-pointed}). Thus, for all $\sigma \in F$, $a_\sigma$ is a pointed cone, as a subset of $N(\phi(\sigma))$ (\Cref{lemma:facet-subset-normal-cone}). We deduce that \Cref{cond:pointed} holds. 

We can now apply \Cref{cor:covering}, and we get that $\bigcup_{\sigma \in F} a_\sigma = \bigcup_{v \in V} \bigcup_{\sigma \in \phi^{-1}(v)} a_\sigma = \R^d$. Since $a_\sigma \subset N(v)$ for all $\sigma \in \phi^{-1}(v)$ and $v \in V$ (\Cref{lemma:facet-subset-normal-cone}), we deduce that $\bigcup_{v \in V} N(v) = \R^d$. This proves that $V \supset \vertices(P)$ by \Cref{prop:vertex-completeness}, and that $P$ is a polytope.
\end{proof}

\begin{proof}[\proofname{} of \Cref{prop:flag}]
We first note that by \Cref{cond:feasibility}, the points $v, w_1, \dots, w_d$ all belong to $P$ and $I_v = I(v)$ and $I(w_k) = I_{w_k}$ for all $k$ (\Cref{lemma:feasibility}).

We claim that the product $M$ of matrices 
\[
M \coloneqq \begin{pmatrix}
\transpose{a_{i_1}} \\ \vdots \\ \transpose{a_{i_d}} 
\end{pmatrix}
\begin{pmatrix}
v - w_1 & \dots & v - w_d  
\end{pmatrix}
\]
is an upper triangular matrix with positive diagonal entries. Indeed, for all $k \in [d]$, $M_{kk} = \scalar{a_{i_k}}{v} - \scalar{a_{i_k}}{w_k} = b_{i_k} - \scalar{a_{i_k}}{w_k} > 0$ as $i_k \notin I(w_k)$ and $i_k \in I(v)$. Similarly, if $j < k$, $M_{kj} = \scalar{a_{i_k}}{v} - \scalar{a_{i_k}}{w_j} = 0$ as $i_k \in I_{w_j} = I(w_j)$. We deduce that $M$ has rank $d$, and so has the matrix $\begin{pmatrix}
a_{i_1} & \dots & a_{i_d} \\
\end{pmatrix}$. This ensures that vectors $a_{i_1}, \dots, a_{i_d}$ are linearly independent. As $v \in P$ and $i_1, \dots, i_d \in I(v)$, we deduce that $v$ is a vertex of $P$. Similarly, we get that the matrix $\begin{pmatrix}
w_1 - v & \dots & w_d - v  
\end{pmatrix}$ has rank $d$. As $v, w_1, \dots, w_d \in P$, we conclude $\dim(P) = d$ (see the characterization of $\dim P$ discussed in \Cref{sec:normal-fan}).
\end{proof}

In the sequel, we suppose that \Cref{cond:feasibility,cond:facet_subset,cond:ridge_count,cond:rank,cond:sep,cond:flag} are met. Recall that $\Ksf$ is the (pure $(d-1)$-dimensional) abstract simplicial complex generated by the elements of $F$. In the following, given $v \in V$, we denote $c_v \coloneqq \sum_{\sigma \in \phi^{-1}(v)} \sigma \in C_{d-1}(\Ksf; \F_2)$ the (mod-$2$) chain of facets $F$ associated with $v$ via the mapping $\phi$. Recall that its boundary is $\partial c_v = \sum_{\sigma \in \phi^{-1}(v)} \partial \sigma$.

\begin{lemma}\label{lemma:cochain-boundary}
Let $v$ be a vertex of $P$. Every $(d-1)$-generic vector $x$ in the relative interior of a facet of $N(v)$ satisfies $x \in a_\tau$ for an odd number of $\tau \in \partial c_v$.
\end{lemma}

\begin{proof}
Let $x$ be a $(d-1)$-generic vector in the relative interior of a facet $F$ of $N(v)$. As a face of $N(v)$, $F$ is the normal cone of some face of $P$, and this face have dimension $1$ (\Cref{prop:normal-cone-basic-prop}\ref{item:i}). Since $P$ is a polytope (by \Cref{th:vertex-completeness}), the edge is a line segment between $v$ and an (adjacent) vertex $w$ of $P$. In consequence, $F = N(v) \cap N(w)$.

We build a generic path from $\interior N(v)$ and $\interior N(w)$ passing through $x$ as follows. The path is linear and directed along vector $y \in \R^d$ that we choose to be ``generic'', in the sense that it is not parallel to any cone $a_\eta$ ($\eta \in \Ksf$) of dimension $d-1$, and not parallel to $F$ (such a choice is possible as these conditions define the complement set of a finite union of hyperplanes of $\R^d$). Given $\eps > 0$, the vectors $z \coloneqq x - \eps y$ and $u \coloneqq x + \eps y$ are $d$-generic vectors lying respectively in $\interior N(v)$ and $\interior N(w)$, up to assuming that $\eps$ is small enough and changing $y$ into $-y$. Moreover, up to taking $\eps$ smaller, the line segments $(x,z]$ and $(x,u]$ do not meet any cone $a_\eta$ ($\eta \in \Ksf$) of dimension less than $d$. In this way, the only element of the path $\gamma \colon t \in [0,1] \mapsto z + t (u - z)$ which is not $d$-generic is the vector $x$, and the latter is $(d-1)$-generic as required. On top of that, the intersection with any cone $a_\eta$ ($\eta \in \Ksf$) of dimension $d-1$ is transversal, owing to the genericity of the direction vector $y$. We deduce that $\gamma$ is a generic path from $z$ to $u$ passing through $x$.

Since \Cref{cond:feasibility,cond:facet_subset,cond:ridge_count,cond:rank,cond:sep} are satisfied, the hypotheses of \Cref{th:odd-covering} hold, as shown in the proof of \Cref{th:vertex-completeness}. Thus, we can proceed as in the proof of \Cref{th:odd-covering}, and we introduce $S(t) \coloneqq \{\sigma \in \facets \colon \; \gamma(t) \in a_\sigma \}$. As shown there, we have
\begin{equation}\label{eq:symdiff2}
S(0) \symdiff S(1) = 
\bigcup_{\tau \in T} \{ \sigma \in \facets \colon \tau \in \partial \sigma \, , \; \dim a_\sigma = d \big\} \, . 
\end{equation}
where $T \coloneqq \{ \tau \in \ridges \colon x \in a_\tau\}$. We claim that $S(0) \subset \phi^{-1}(v)$ and $S(1) \subset \phi^{-1}(w)$, so that both sets are disjoint. Indeed, let $\sigma \in S(0)$. Then, $a_\sigma \subset N(\phi(\sigma))$, and thus $z \in (\interior N(v)) \cap N(\phi(\sigma))$, which implies $\phi(\sigma) = v$ by \Cref{cor:normal-cone-basic-prop}. Symmetrically, we can show $S(1) \subset \phi^{-1}(w)$. 

We deduce that 
\begin{equation}\label{eq:S0}
S(0) = \bigcup_{\tau \in T} \{ \sigma \in \phi^{-1}(v) \colon \tau \in \partial \sigma \, , \; \dim a_\sigma = d \big\}
\end{equation}
As shown in the proof of \Cref{th:odd-covering}, the union in~\eqref{eq:symdiff2}, and subsequently  in~\eqref{eq:S0}, is disjoint. We deduce that
\begin{align*}
& \# \{\tau \in \partial c_v \colon x \in a_\tau\} \\
& = \sum_{\substack{\tau \in \ridges\\ x \in a_\tau}} \# \{ \sigma \in \phi^{-1}(v) \colon \tau \in \partial \sigma \} \mod 2 \\
& = 
\sum_{\substack{\tau \in \ridges\\ x \in a_\tau}} \# \{ \sigma \in \phi^{-1}(v) \colon \tau \in \partial \sigma \, , \; \dim a_\sigma = d \}  \\
& + 
\sum_{\substack{\tau \in \ridges\\ x \in a_\tau}} \# \{ \sigma \in \phi^{-1}(v) \colon \tau \in \partial \sigma \, , \; \dim a_\sigma < d \} \mod 2 \\
& = \# S(0) + \sum_{\substack{\sigma \in \phi^{-1}(v) \\ \dim a_\sigma < d}} \# \{ \tau \in \partial \sigma \colon x \in a_\tau \} \mod 2 
\reasonnext{as the union in~\eqref{eq:S0} is disjoint,}\\[\jot]
& = \# S(0) \mod 2  
\reason{\begin{minipage}[t]{0.25\textwidth}
\begin{flushright}
by \Cref{lemma:non-full-dim-cones} \\
and the pointedness of the $a_\sigma$,
\end{flushright}
\end{minipage}} \\[\jot]
& = 1 \mod 2  
\reason{%
by \Cref{th:odd-covering} 
and the $d$-genericity of $z$. \qedhere} 
\end{align*}
\end{proof}

\begin{lemma}\label{lemma:boundary}
Let $v \in V$. For all $\tau \in \partial c_v$, there is a pair $(\sigma, \sigma') \in \phi^{-1}(v) \times (F \setminus \phi^{-1}(v))$ of facets adjacent in $G$ such that $\tau = \sigma \cap \sigma'$. 
\end{lemma}

\begin{proof}
We proceed analogously to the proof of \Cref{lemma:ridge-even} and label every edge of $G$ between $\sigma$ and $\sigma'$ by $\sigma \cap \sigma'$. 

Let $v \in V$ and $\tau \in \ridges$. As shown in the proof of \Cref{lemma:ridge-even}, the set $S = \{ \sigma \in \facets \colon \tau \subset \sigma \}$ is precisely the set of vertices of the edges of $G$ labeled by $\tau$. Thus, the set $S_v \coloneqq \{ \sigma \in \phi^{-1}(v) \colon \tau \subset \sigma \}$ is precisely the set of the vertices of the edges of $G$ labeled by $\tau$ that belong to $\phi^{-1}(v)$. Moreover, the set of edges labeled by $\tau$ in $G$ form a matching, as shown in the proof of \Cref{lemma:ridge-even}. Let $M_v$ be the submatching formed by the edges incident to $\phi^{-1}(v)$. If the vertices of $M_v$ all belong to $\phi^{-1}(v)$, then $S_v$ coincides with the vertices of $M_v$. Thus, $S_v$ has even cardinality, and $\tau \notin \partial c_v$.
\end{proof}

\begin{proof}[\proofname{} of \Cref{th:graph-completeness}]
We already know that $\phi(F) \subset V = \vertices(P)$ by \Cref{cor:vertex-set}. We claim that $\vertices(P) \subset \phi(F)$. Let $v$ be a vertex of $P$, and $c \in \R^d$ such that $v$ is the unique maximizer of the linear map $x \mapsto \scalar{c}{x}$ over $P$. As shown in the proof of \Cref{th:vertex-completeness}, $\bigcup_{w \in V} \bigcup_{\sigma \in \phi^{-1}(w)} a_\sigma = \R^d$. Let $w \in V$ and $\sigma \in \phi^{-1}(w)$ such that $c \in a_\sigma$. As $a_\sigma \subset N(w)$ (\Cref{lemma:facet-subset-normal-cone}) and $c \notin N(w)$ for all $w \in P$, $w \neq v$, we deduce that $w = v$. Thus, $v = \phi(\sigma) \in \phi(F)$.

Now, let $v, w$ be two adjacent vertices in $G_P$. The normal cone $N([v,w])$ of $[v,w]$ is a facet of the normal cone of $v$ (\Cref{prop:normal-cone-basic-prop}). Let $u \in \relint(N([v,w]))$ be a $(d-1)$-generic vector (the existence of $u$ is guaranteed because the dimension of $N([v,w])$ is $d-1$). By \Cref{lemma:cochain-boundary}, we deduce that there exists $\tau \in \partial c_v$ such that $u \in a_\tau$. By \Cref{lemma:boundary}, there is a pair $(\sigma, \sigma') \in \phi^{-1}(v) \times (F \setminus \phi^{-1}(v))$ of facets adjacent in $G$ such that $\tau = \sigma \cap \sigma'$. Let $x \coloneqq \phi(\sigma')$. By definition, $v$ and $x$ are adjacent in the graph $\phi(G)$. We observe that $a_\tau \subset a_{\sigma'} \subset N(x)$ since $\tau \subset \sigma'$ (\Cref{lemma:facet-subset-normal-cone}). Thus, $N(x)$ contains $u \in \relint(N([v, w]))$, which implies that $x$ belongs to the edge $[v,w]$ (\Cref{cor:normal-cone-basic-prop}). As, by \Cref{cor:vertex-set}, $x \in V$ is a vertex of $P$, it is a vertex of the edge $[v,w]$. Thus, $x \in \{v,w\}$. We cannot have $x = v$ as $\sigma' \notin \phi^{-1}(v)$. We deduce that $x = w$. As a result, $v$ and $w$ are adjacent in $\phi(G)$.
\end{proof}

\begin{proof}[\proofname{} of \Cref{th:graph}]
Throughout the proof, we use the fact that $I(v) = I_v$ for all $v \in V$ (as a consequence of \Cref{lemma:feasibility}).

Let $v \in V$ be a vertex of $P$. By \Cref{th:graph-completeness}, we already know that all vertices of $P$ adjacent to $v$ in $G(P)$ are also adjacent to $v$ in $\phi(G)$. Conversely, let $w$ be a point of $V$ adjacent to $v$ in $\phi(G)$. Let $Q$ be the minimal face of $P$ containing the line segment $[v,w]$. The point $v$ is a vertex of $Q$ (since it is a vertex of $P$ contained in $Q$). Since $Q$ is not reduced to $\{v\}$ (we have $w \neq v$), we can consider an edge $E = [v,w']$ of $Q$ incident to $v$. The point $w'$ is a vertex of $P$, and $[v, w']$ is also an edge of $P$ (faces of $Q$ are faces of $P$). Moreover, by \Cref{th:graph-completeness}, $w'$ is adjacent to $v$ in $\phi(G)$. Note that $I_v \cap I_w = I(Q)$, as $Q$ is the minimal face the segment $[v,w]$ and  the set of inequalities activated by all points of $[v,w]$ is $I_v \cap I_w$. As $Q$ contains the edge $[v,w']$, we have $I(Q) \subset I([v,w']) = I_v \cap I_{w'}$. Thus, $I_v \setminus I_w \supset I_v \setminus I_{w'}$. By \Cref{cond:incomparable-edges}, we deduce that $w = w'$. We conclude that $w$ is a vertex of $P$ adjacent to $v$.  
\end{proof}

\begin{proof}[\proofname{} of \Cref{prop:simplicial-refinement}]
By abuse of notation, given a point $\tilde x$ of $\tilde P$, we denote by $I(\tilde x) \coloneqq \forall i \in [n] \, , \; \scalar{a_i}{x} \leq b_i + \eps^{i} \}$ the set of active inequalities at $\tilde x$ in the inequality system defining $\tilde P$.

Let $\tilde v$ be a vertex of $\tilde P$. The vectors $a_i$ for $i \in I(\tilde v)$ are $d$ linearly independent vectors. Thus, the normal cone $N(\tilde v) = \cone{\{a_i\}_{i \in I(\tilde v)}}$ is a simplicial cone. This proves that the normal fan of $\tilde P$ is a simplicial fan.

In order to prove the refinement property, it suffices to prove to consider maximal cones in the normal fan of $\tilde P$, \ie, normal cones of vertices of $\tilde P$. Let $\tilde v$ be a vertex of $\tilde P$. As explained earlier, there is a vertex $v$ of $P$ such that $I(\tilde v) \subset I(v)$. By the characterization~\eqref{eq:normal-cone-characterization} applied to $P$ and $\tilde P$, we deduce that the normal cone of $\tilde v$ (w.r.t.~$\tilde P$) is given by $\cone{\{a_i\}_{i \in I(\tilde v)}}$, while that of $v$ (w.r.t.~$P$) by $\cone{\{a_i\}_{i \in I(v)}}$. We deduce that the former is included in the latter.
\end{proof}

\begin{proof}[\proofname{} of \Cref{prop:t1-to-t3-complete}]
\Cref{cond:feasibility,cond:facet_subset} trivially hold by construction. As shown in~\cite{AllamigeonCanuStrub2023}, the graph $G$ is $d$-regular, \ie, every lex-feasible basis has precisely $d$ neighbors in the adjacency graph $G$. Let $I$ be a lex-feasible basis, and $i \in I$. The fact that there exists an adjacent lex-feasible basis $I'$ such that $I \setminus \{i\} \subset I'$ can be shown constructively by applying one pivot step the simplex method from the basis $I$ of the perturbed polytope $\tilde P$; see~\cite[Chapter~11]{Schrijver86}. This proves \Cref{cond:ridge_count}.
\end{proof}

\begin{proof}[\proofname{} of \Cref{prop:t4-t5-complete}]
\Cref{cond:rank} has already been discussed, so we focus on \Cref{cond:sep}. First, as $\tilde x^{(\sigma)}$ belongs to $\tilde P$, we have $\scalar{a_i}{x^{(\sigma)}} \leq b_i + \eps^{i}$. This ensures that $s_{\sigma} \in \R_{\geq 0}^d$. Let $x^{(\sigma^\star)}$ be the vertex of $\tilde P$ corresponding to the lex-feasible basis $\sigma^\star$. For all $j \in [d]$, we have $\scalar{a_{i_j}}{x^{(\sigma^\star)}} = b_{i_j} + \eps^{i_j}$, and so 
\[
\sum_k (s_\sigma)_k f_j = \Big(\sum_j f_j \transpose{a_{i_j}}\Big) (\tilde x^{(\sigma^\star)} - \tilde x^{(\sigma)}) = \tilde x^{(\sigma^\star)} - \tilde x^{(\sigma)} \, .
\]
Let $i \in \sigma$. We get $\scalar[\big]{a_i}{\sum_j (s_\sigma)_j f_j} = \scalar{a_i}{\tilde x^{(\sigma^\star)} - \tilde x^{(\sigma)}} = \scalar{a_i}{\tilde x^{(\sigma^\star)}} - (b_i + \eps^{i}) \leq 0$ as $\scalar{a_i}{\tilde x^{(\sigma)}} = b_i + \eps^{i}$ and $\tilde x^{(\sigma^\star)}$ belongs to $P$. Finally, if $s_\sigma = 0$, we would have $\scalar{a_{i_j}}{x^{(\sigma)}} = b_{i_j} + \eps^{i_j}$, which would imply $x^{(\sigma^\star)} = x^{(\sigma)}$ (as the $a_{i_j}$ are linearly independent). This is impossible because lex-feasible bases are not degenerate~\cite{Avis2000}.
\end{proof}

\section{Description of the Instances}\label{app:experiments}

The instances used for the experiments are the following:
\begin{itemize}
\item the $n$-Birkhoff polytope (instance \C$birkhoff_$$n$) is the polytope of $n\times n$ doubly stochastic matrices. More precisely, it is the set of points $x\in\R_{\geq 0}^{n\times n}$ such that $\sum_{j=1}^n x_{ij}=1$ for every $i\in[n]$ and $\sum_{i=1}^n x_{ij}=1$ for every $j\in[n]$. Its dimension is $(n-1)^2$, and its vertices are precisely the $n!$ permutation matrices.
\item the $d$-dimensional cross-polytope (instance \C$cross_$$d$) is the convex hull of the $2d$ points ${\pm e_i : i\in[d]}$, where $(e_i)_{i\in[d]}$ denotes the canonical basis of $\R^d$. Equivalently, this is the unit ball for norm $\| \cdot \|_1$. It is defined by the $2^d$ inequalities $\sum_{i = 1}^d \pm x_i \leq 1$. It has dimension $d$ and $2d$ vertices. 
\item the $d$-dimensional cube (instance \C$cube_$$d$) is the polytope $[0,1]^d$. It has dimension $d$ and $2^d$ vertices, corresponding to the vectors of ${0,1}^d$.
\item the $(d,n)$-dual cyclic polytope (instance \C$cyclic_$$d$\C$_$$n$) is defined by means of the moment curve $u(t)=(t,t^2,\ldots,t^d)$. This is a $d$-dimensional polytope given by $n$ inequalities $\scalar{u(t_i)-c}{x} \leq 1$ where $t_1 < t_2 < \dots < t_n$ are fixed parameters, and $c = \frac{1}{n} \sum_{t = 1}^n u(t_i)$. The $(d,n)$-dual cyclic polytope is known to maximize the number of vertices among $d$-dimensional polytopes with $n$ facets. Here, we focus on the case where $n = 2d$, where the number of vertices grows exponentially in $d$.

\item the $(k,n)$-hypersimplex is defined as the set of points $x \in [0,1]^n$ satisfying $\sum_{i=1}^n x_i = k$. For $0 < k < n$, it has dimension $n-1$, and its vertices are the $\binom{n}{k}$ vectors of $\{0,1\}^n$ having exactly $k$ nonzero coordinates. The instance \C$hypersimplex_$$d$ stands for the $(k,n)$-hypersimplex with $k = \big\lfloor\frac{d+1}{2}\big\rfloor$ and $n = d+1$.
\item the permutahedron of order $n$ (instance \C$permutahedron_$$n$) is the convex hull of all permutations of the vector $(1,\dots,n)$. It lies in the affine hyperplane $\sum_{i=1}^n x_i=n(n+1)/2$, and therefore has dimension $n-1$. Its $n!$ vertices are precisely the permutations of $(1,\dots,n)$.
\item the instances \C$poly20dim21$ and \C$poly23dim24$ are the two counterexamples to the Hirsch conjecture from~\cite{Matschke2015}. They are respectively $20$- and $23$-dimensional polytopes defined by $40$ and $46$ inequalities involving very large rationals.
\end{itemize}

\end{document}